\documentclass[journal,12pt,draftcls,onecolumn]{IEEEtran}
\IEEEoverridecommandlockouts

\usepackage{amsthm}
\usepackage{amssymb}
\usepackage{amsmath}
\usepackage{amssymb}
\usepackage{epsfig}
\usepackage{epsf}
\usepackage{subfigure}
\usepackage{graphicx}
\usepackage{cite}
\usepackage{url}
\usepackage[]{authblk}
\usepackage{color}
\usepackage{upgreek}
\usepackage{comment}

\usepackage{wrapfig}

\usepackage{pifont}% http://ctan.org/pkg/pifont
\newcommand{\cmark}{\ding{51}}%
\newcommand{\xmark}{\ding{55}}%

\def\BibTeX{{\rm B\kern-.05em{\sc i\kern-.025em b}\kern-.08em
    T\kern-.1667em\lower.7ex\hbox{E}\kern-.125emX}}

\newtheorem{corollary}{Corollary}
\newtheorem{theorem}{Theorem}

\newtheorem{lemma}{Lemma}

\newtheorem{remark}{Remark}

\newcommand{\msf}{\mathsf}
\newcommand{\lbp}{\left\{}
\newcommand{\rbp}{\right\}}
\newcommand{\lp}{\left(}
\newcommand{\rp}{\right)}
\begin{document}

%\title{Index Coding Over Broadcast Erasure Channels}
%\title{Erasure Broadcast Channels with\\ Random Receiver Side-Information}
\title{Content Delivery over Broadcast Erasure Channels with Distributed Random Cache}

% \title{Impact of side of Feedback in Broadcast Channels with Random Receiver Side-Information}

\author{
Alireza~Vahid, Shih-Chun~Lin, I-Hsiang~Wang, Yi-Chun Lai
\thanks{Alireza Vahid is with the Electrical Engineering Department of the University of Colorado Denver, Denver, CO, USA. Email: {\sffamily alireza.vahid@ucdenver.edu}. Shih-Chun Lin and Yi-Chun Lai are with Department of Electrical and Computer Engineering, NTUST, Taipei, Taiwan. Email: {\sffamily sclin@ntust.edu.tw}. I-Hsiang Wang is with Department of Electrical Engineering, National Taiwan University, Taipei, Taiwan. Email: {\sffamily ihwang@ntu.edu.tw}.}
}

\maketitle

%%%%%%%%%%%%%%%%%%%%%%%%%%%%%%%%%%%%%%%%%%%%%%%%%%%

\begin{abstract}
We study the content delivery problem between a transmitter and two receivers through erasure links,  when each receiver has access to some random side-information about the files requested by the other user. The random side-information is cached at the receiver via the decentralized content placement. The distributed nature of receiving terminals may also make the erasure state of two links and indexes of the cached bits not perfectly known at the transmitter. We thus investigate the capacity gain due to various levels of availability of channel state and cache index information at the transmitter. More precisely, we cover a wide range of settings from global delayed channel state knowledge and a non-blind transmitter (i.e. one that knows the exact cache index information at each receiver) all the way to no channel state information and a blind transmitter (i.e. one that only statistically knows cache index information at the receivers). We derive new inner and outer bounds for the problem under various settings and provide the conditions under which the two match and the capacity region is characterized. Surprisingly, for some interesting cases the capacity regions are the same even with single-user channel state or single-user cache index information at the transmitter.
% We consider two settings: (1) blind transmitter where only the statistics of the side-information is known to the transmitter, and (2) non-blind transmitter where the transmitter knows exactly what each receiver has access to. We characterize the capacity region in the non-blind setting, while under certain conditions, we provide matching inner and outer-bounds for the blind setting. We show how to capture the availability of random receiver side-information in the outer-bounds, and how to leverage the exact or statistical knowledge of the side-information at the transmitter for communications.
\end{abstract}

%\begin{IEEEkeywords}
%Broadcast erasure channels, index coding, caching, capacity region, feedback.
%\end{IEEEkeywords}

%%%%%%%%%%%%%%%%%%%%%%%%%%%%%%%%%%%%%%%%%%%%%%%%%%%

\section{Introduction}
\label{Section:Introduction}

Available receiver-end side-information can greatly enhance content delivery and increase the attainable data rates in wireless systems. In particular, in various applications such as caching~\cite{maddah2015cache,naderializadeh2017fundamental}, coded computing~\cite{prakash2018coded}, private information retrieval~\cite{chor1995private,sun2017private}, and index coding~\cite{bar2011index,chaudhry2008efficient,maleki2014index}, side-information is intentionally and strategically placed at each receiver's cache during some placement phase in order to lighten future communication loads. However, in a more realistic setting, there is no centralized mechanism to populate the local caches. On the other hand, receivers may obtain some side-information by simply over-hearing the signals intended for other nodes over the shared wireless medium. The transmitter(s) may or may not be aware of the exact content of the side-information at each receiver, and if the transmitter is not aware of the exact content, the problem is referred to as Blind Index Coding~\cite{kao2016blind}. Further, the transmitter(s) may be enhanced by receiving channel state feedback. This paper focuses on content delivery in  wireless networks with potential channel state feedback and/or random available side-information at the receivers.

In a packet-based communication network, instead of the classic Gaussian channel, the network-coding-based approaches generally model each communication hop as a packet erasure channel. In this work, as \cite{ghorbel2016content}, we focus on broadcast erasure channels with random receiver side-information. More specifically, we consider a single transmitter communicating with two receiving terminals through erasure links. Through decentralized content placement \cite{kao2016blind}\cite{ghorbel2016content}, each receiver has randomly cached some side-information about the message bits (file) of the other user. In \cite{ghorbel2016content},  for both users, their indexes of cached bits during the placement and (delayed) channel erasure state during the delivery are globally known at the transmitter. However, as pointed out in \cite{DistributedICShamai}\cite{lin2019no}, in distributed networks, acquiring such global information at transmitters is prohibitive due to the extensive date exchange and the heterogenous capabilities for receiving terminals to feed back. Thus we consider three scenarios about the transmitter's knowledge of the receiver-end side-information: {(1) the blind-transmitter case where only the statistics of the \emph{random} cache index information is known to the transmitter, (2) the non-blind-transmitter case where the transmitter knows exactly what each receiver has access to, and (3) the semi-blind-transmitter case that knows the exact cache index information at only one of the receivers}. Meanwhile, a range of assumptions on the availability of channel state information at the transmitter (CSIT) are also considered: (1) when both receivers provide delayed CSI to the transmitter (DD); (2) when only one receiver provides delayed CSI to the transmitter (DN); and (3) when receivers do not provide any CSI to the transmitter (NN).

The blind-transmitter case under scenario NN in \cite{kao2016blind} corresponds to the setting that during the whole content placement and delivery, both receivers can not have capabilities to feed back through the control channel or shared wireless medium. Then the transmitter has neither cache index information nor CSI. On the other hand, for each receiver, if the feedback resource is available for some period of time and then becomes unavailable or intermittent~\cite{IFB-ITW,IntermittentIHsiang}, then we have other scenarios. For example, though during the placement both receivers can feed back cache indexes, during the delivery one of the receiver may not be able to feed back the fast varying channel state of its own link. Then we have a non-blind transmitter under scenario DN.  This setting could also arise due to new security threats that aim to disrupt the flow of control packets in distributed systems~\cite{liu2020adversarial,sadeghi2018adversarial,sadeghi2019physical}, that is, during the delivery the CSI feedback from one of the receivers is severely attacked.

%Finally, one could consider a hybrid case in which the transmitter knows part of the side-information at each receiver and about the remaining part only has statistical knowledge. We do not study this final scenario in this work.

\noindent \underline{\bf Related Work and Literature Review}: To better place our work within the literature for cached network or index coding, we summarize some of the main results {in the literature}. The classic index coding problem~\cite{birk2006coding,bar2011index,alon2008broadcasting} assumes the transmitter is aware of the exact side-information at each node~\cite{chaudhry2008efficient,ong2014linear}, and the network does not have wireless links directly. This problem has been a powerful tool in studying the data communication network with receiver caching~\cite{jafar2013topological,maleki2014index,maddah2014fundamental}.

For index coding over wireless links, as aforementioned, two extreme cases have been studied \cite{kao2016blind}\cite{ghorbel2016content}. In~\cite{ghorbel2016content}, the capacity regions for two and three-user broadcast erasure channels, when the transmitter has access to global delayed CSI and cached index information, are characterized. This setting assumes rather strong and stable feedback channels from all receivers. On the other hand,  blind wireless index coding is introduced in~\cite{kao2016blind} where the transmitter only knows the statistics of the CSI and cached index information since there is no feed back at all from the two users. However, a limitation is added in \cite{kao2016blind} such that only the user with weaker link (higher erasure probability) has cache. Even under this limitation, the capacity region is not fully known In conclusion, there remains a gap in the index coding literature between these two extreme points when it comes to wireless setting, which is the target of our paper. More comparisons with \cite{kao2016blind}\cite{ghorbel2016content} can be found in Section \ref{Section:Main_BIC}.

\noindent \underline{\bf Contributions}: 
{In this work, we consider the two-user erasure broadcast channel with the erasure states of the two links being independent of each other.} We first present the capacity for the two-user broadcast erasure channel with a non-blind transmitter under scenario NN, that is, no CSI feedback. For scenario NN, we observe that the stronger receiver, \emph{i.e.} the one whose channel has a smaller erasure probability, will eventually be able to decode messages intended for both receivers. Note that this observation does not imply that the broadcast channel is degraded due to the cache.  Interestingly, the achievability derived from this observation indicates that our optimal protocol works even when the transmitter does not know the exact cache index information at the stronger receiver, \emph{i.e.} with a ``semi-blind'' transmitter. To derive the outer-bounds, besides using the aforementioned observations, we also show an extremal entropy inequality between the two receivers that captures the availability of receiver-end side-information. With a blind transmitter, these outer-bounds can be automatically applied since the transmitter has more knowledge in the non-blind case, and {we show that the outer bounds} can be achieved when the channel is symmetric.  When some delayed CSI is available, we also demonstrate the capacity with a non-blind transmitter and global CSI (scenario DD) \cite{ghorbel2016content} can be achieved when the transmitter has less information. In particular, for the blind-transmitter case under scenario DD, we provide an optimal protocol for the symmetric channel. For the non-blind-transmitter case or semi-blind-transmitter case, we also extend our earlier non-cached capacity result for scenario DN \cite{lin2019no}, and show that even with cache the capacity region with only single-user CSI feed back can match that with global CSI in \cite{ghorbel2016content}.

%In the achievability strategy, the first step is to communicate the message intended for the weaker receiver. The stronger receiver will be able to decode this message faster than the intended receiver and thus, in the second step, we include the part of the message for the stronger user that is available at the weaker receiver. During the final step, the remaining part of the message intended for the stronger receiver is delivered.

%The second set is for the two-user broadcast erasure channels with delayed CSI feedback, random side-information at the receivers, and %with non-blind transmitter. We note that these bounds could be obtained independently from the results of~\cite{ghorbel2016content}, %but we provide an alternative proof. The derivation involves modifying the proof of the no-CSI case to include delayed feedback.

\begin{table*}[t]
\begin {center}
 \vspace{-2mm}
 \caption{{Summary of Contributions Categorized by Transmitter's Information}}
 \vspace{-4mm}
 \label{Table:Summary}
 \begin{tabular}{|c|c|c|c|c|c|c|c|c|c|}
 \hline
 \multicolumn{2}{ |c| }{State} & \multicolumn{2}{ |c| }{Cache} & \qquad \qquad  Contributions &\multicolumn{2}{ |c| }{State} & \multicolumn{2}{ |c| }{Cache} & \qquad \qquad  Contributions\\
 W & S & W & S & & W & S & W & S &\\
 \hline
 \hline
 \xmark & \xmark & \cmark & \cmark & *$\mathcal{C}^{\mathrm{non-blind}}_{\mathrm{NN}}$  in Theorem~\ref{THM:Capacity_Out_BIC_No}
  & \cmark & \xmark &   \xmark & \cmark &$\mathcal{C}^{\mathrm{semi-blind}}_{\mathrm{DN}}$ in Theorem~\ref{THM:Capacity_Out_BIC_Delayed} Case B\\

 \hline
 \xmark & \xmark & \cmark &  \xmark & *$\mathcal{C}^{\mathrm{semi-blind}}_{\mathrm{NN}}$ in Corollary ~\ref{remark:semiblind}  & \xmark & \cmark & \xmark & \xmark & Remains open\\
 \hline
 \xmark & \xmark & \xmark & \cmark  & Inner-bounds in Theorem~\ref{THM:Blind-Ach} & \cmark & \cmark & \cmark & \cmark & $\mathcal{C}^{\mathrm{non-blind}}_{\mathrm{DD}}$ previously known~\cite{ghorbel2016content}\\
 \hline

 \xmark & \xmark & \xmark & \xmark & Symmetric $\mathcal{C}^{\mathrm{blind}}_{\mathrm{NN}}$ in Theorem~\ref{THM:Blind} & \cmark & \cmark & \cmark & \xmark & $\mathcal{C}^{\mathrm{semi-blind}}_{\mathrm{DD}}$ in Theorem~\ref{THM:Capacity_Out_BIC_Delayed} Case B \\
  \hline
  \xmark & \cmark & \cmark & \cmark & *$\mathcal{C}^{\mathrm{non-blind}}_{\mathrm{DN}}$ in Theorem~\ref{THM:Capacity_Out_BIC_Delayed} Case C & \cmark & \cmark & \xmark & \cmark & $\mathcal{C}^{\mathrm{semi-blind}}_{\mathrm{DD}}$ in Theorem~\ref{THM:Capacity_Out_BIC_Delayed} Case B \\
 \hline
   \xmark & \cmark & \cmark & \xmark & $\mathcal{C}^{\mathrm{semi-blind}}_{\mathrm{DN}}$ in Theorem~\ref{THM:Capacity_Out_BIC_Delayed} Case B & \cmark & \cmark & \xmark & \xmark  & Symmetric $\mathcal{C}^{\mathrm{blind}}_{\mathrm{DD}}$ in Theorem~\ref{THM:Capacity_Out_BIC_Delayed} Case A  \\
\hline
\end{tabular}
\end{center}
\end{table*}

Table~\ref{Table:Summary} summarizes our contributions, which we will further elaborate upon in Section~\ref{Section:Main_BIC}. In this table, ``W'' stands for the weak user (the one with higher erasure probability) and ``S'' stands for the strong user. Further, ``State" columns are for CSIT, and ``Cache" columns are for the transmitter's knowledge of bit indices cached at the receiver: ``\cmark'' indicates availability and ``\xmark'' indicates missing information. For example, ``\xmark \xmark \cmark \xmark'' is the no CSI case (NN) with semi-blind Tx that knows the cache index information at the weaker receiver. For the three cases marked with ``*", the capacity regions are fully identified without additional limitations on system parameters.

% Finally, for the blind transmitter, we show that the outer-bounds that were derived for the non-blind transmitter can be achieved with blind transmitter when: (1) the weaker receiver has no or full side-information, or (2) when channel parameters are symmetric. We note that the blind transmitter model studied here is also referred to as Blind Index Coding over wireless channel~\cite{kao2016blind}. Compared to~\cite{kao2016blind}, we provide more general outer-bounds and achieve these bounds in more cases.

\noindent \underline{\bf Paper Organization}: The rest of the paper is organized as follows. In Section~\ref{Section:Problem_BIC}, we present the problem setting and the assumptions we make in this work. Section~\ref{Section:Main_BIC} presents the main contributions and provides further insights and interpretations of the results. The proof of the main results will be presented in the following sections. Finally, Section~\ref{Section:Conclusion_BIC} concludes the paper.

%%%%%%%%%%%%%%%%%%%%%%%%%%%%%%%%%%%%%%%%%%%%%%%%%%%

\section{Problem Formulation}
\label{Section:Problem_BIC}

Following \cite{ghorbel2016content}, we consider the canonical two-user broadcast erasure channel in Figure~\ref{Fig:BC-BIC} to understand how transmitters can exploit the available side-information at the receivers to improve the capacity region. In this network, a transmitter, $\msf{Tx}$, wishes to transmit two independent messages (files), $W_1$ and $W_2$, to two receiving terminals $\msf{Rx}_1$ and $\msf{Rx}_2$, respectively, over $n$ channel uses. Each message, $W_i$, contains $m_i$ data packets (or bits) which we denote by $\vec{a} = \left( a_1,a_2,\ldots,a_{m_1} \right)$ for $\msf{Rx}_1$, and by $\vec{b} = \left(b_1,b_2,\ldots,b_{m_2} \right)$ for $\msf{Rx}_2$. Here, we note that each packet is a collection of encoded bits, however, for simplicity and without loss of generality, we assume each packet is in the binary field, and we refer to them as bits. Extensions to broadcast packet erasure channels where packets are in large finite fields are straightforward as in~\cite{ghorbel2016content}\cite{vahid2014communication}.

\noindent \underline{\bf Channel model:} At time instant $t$, the messages are mapped to channel input $X[t] \in \mathbb{F}_2$, and the corresponding received signals at $\msf{Rx}_1$ and $\msf{Rx}_2$ are
\begin{align}
\label{eq_DL_channel}
Y_1[t] = S_1[t] X[t]~~ \; \mbox{and} \;~~ Y_2[t] = S_2[t] X[t],
\end{align}
respectively, where $\lbp S_i[t]\rbp$ denotes the Bernoulli $(1-\delta_i)$ process that governs the erasure at $\mathsf{Rx}_i$, and is independently and identically distributed (i.i.d.) over time and across users. When $S_i[t]=1$, $\mathsf{Rx}_i$ receives $X[t]$ noiselessly; and when $S_i[t]=0$, it receives an erasure. In other words, as we assume receivers are aware of their local channel state information, each receiver can map the received signal when $S_i[t]=0$ to an erasure.

\noindent \underline{\bf CSI assumptions:} We assume the receivers are aware of the channel state information (\emph{i.e.} global CSIR). For the transmitter, on the other hand, we assume the following scenarios:
\begin{enumerate}

\item NN or No CSIT model: The transmitter knows only the erasure probabilities and not the actual channel realizations;

\item DN model: The transmitter knows the erasure probabilities and the actual channel realizations of one receiver with unit delay;

\item DD or delayed CSIT model: The transmitter knows the erasure probabilities and the actual channel realizations of both receivers with unit delay.

\end{enumerate}

%We also assume $\delta_{12}=P\{S_1[t]=0,S_2[t]=0\}$.

\begin{figure}[!ht]
\centering
\includegraphics[width = 0.5\columnwidth]{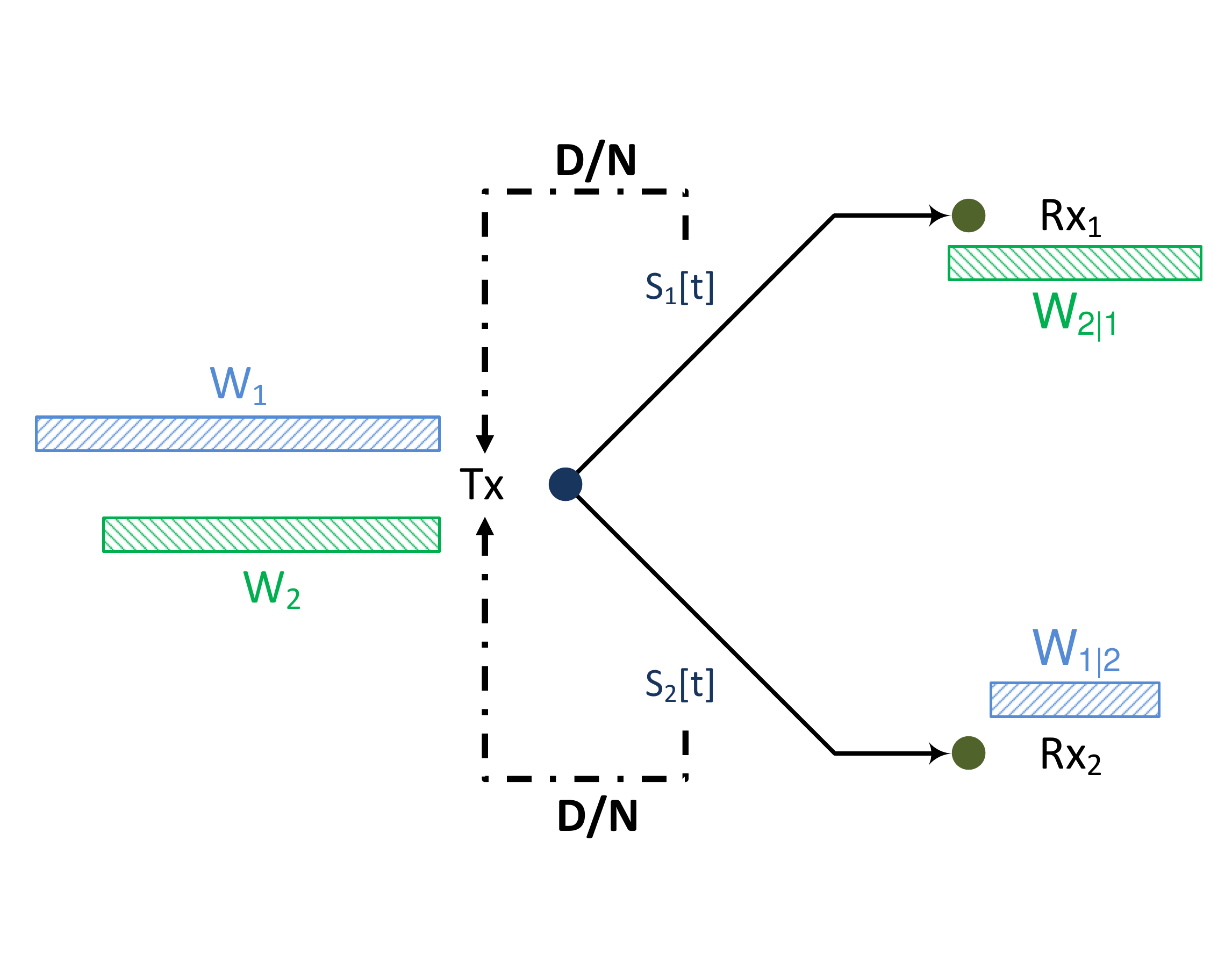}
\caption{Two-user broadcast erasure channels with channel state feedback and random side-information at each receiver.\label{Fig:BC-BIC}}
\end{figure}

\noindent \underline{\bf Available receiver side-information:} Decentralized content placement \cite{kao2016blind}\cite{ghorbel2016content} is adopted, where each user independently caches a subset of the message bits (file). In particular, we assume a random fraction $(1-\epsilon_{i})$ of the bits intended for receiver $\msf{Rx}_{\bar{i}}$ is cached at $\msf{Rx}_{{i}}$, $\bar{i} \overset{\triangle}= 3 - i$, and we denote this side information with $W_{\bar{i}|i}$ as in Figure~\ref{Fig:BC-BIC}. This assumption on the available side-information at each receiver could also be represented using an erasure side channel. More precisely, we can assume available side-information to receiver $i$ is created through
\begin{align}
E_i[\ell]W_{\bar{i}}[\ell], \qquad \ell = 1,2,\ldots, nR_{\bar{i}}, \label{eq_cacheModel}
\end{align}
where $W_{\bar{i}}[\ell]$ is the $\ell^\mathrm{th}$ bit of message $W_{\bar{i}}$, while the cache index information $E_i[\ell]$ is an i.i.d. Bernoulli $(1-\epsilon_i)$ process independent of all other channel parameters and known at receiver $i$. To concisely present our result, in our placement model each user does not cache its own message but only the interference. Our results can be easily extended to the case when both the own message and interference bits are cached.\\

%and that
%\begin{align}
%\label{Eq:MessageSplitting}
%H\lp W_{\bar{i}|i} \rp = \lp 1 - \epsilon_i \rp H\lp W_{\bar{i}} \rp, \qquad i=1,2.
%\end{align}

\noindent \underline{\bf  The Transmitter's knowledge of the cache index:} Following the convention which presents the length $nR_{\bar{i}}$ sequence $E_i[1], \ldots E_i[nR_{\bar{i}}]$ as $E^{nR_{\bar{i}}}_i$, we consider three scenarios for the blindness of the cache index at the transmitter
\begin{enumerate}

\item Blind Transmitter: In this scenario, the transmitter's knowledge of the receiver side-information is limited to the values of $\epsilon_1$ and $\epsilon_2$, while the cache index information $E^{nR_2}_1$ and $E^{nR_1}_2$ in \eqref{eq_cacheModel} is unknown.

\item Semi-Blind Transmitter: In this scenario, the transmitter's knowledge of the side-information at $\msf{Rx}_i$ is limited to the value of $\epsilon_i$, while the transmitter knows $W_{i|\bar{i}}$ through $E^{nR_{i}}_{\bar{i}}$

\item Non-Blind Transmitter: In this scenario, the transmitter knows exactly what fraction of each message is available to the unintended receiver. In other words, the transmitter knows $W_{2|1}$ and $W_{1|2}$ through the the cache index information $E^{nR_2}_1$ and $E^{nR_1}_2$.\\
\end{enumerate}

\noindent \underline{\bf Encoding:} We start with the NN model where the constraint imposed at the encoding function $f_t(.)$ at time index $t$ for the blind scenario is
\begin{align}
\label{eq_enc_function_blind}
X[t] = f_t\lp W_1, W_2, \msf{PI} \rp,
\end{align}
for the non-blind scenario is
\begin{align}
\label{eq_enc_function_nonblind}
X[t] = f_t\lp W_{1|2}, \bar{W}_{1|2}, W_{2|1}, \bar{W}_{2|1}, \msf{PI} \rp,
\end{align}
and
\begin{align}
\label{eq_enc_function_semiblind}
X[t] = f_t\lp W_{i|\bar{i}}, \bar{W}_{i|\bar{i}}, W_{\bar{i}}, \msf{PI} \rp,
\end{align}
where $\msf{PI}$ represents the knowledge of statistical parameters $\delta_1, \delta_2, \epsilon_1,$ and $\epsilon_2$, and $\bar{W}_{i|\bar{i}}$ is the complement of $W_{i|\bar{i}}$ with respect to $W_i$, $i=1,2$. Although not ideal, this notation is adopted to highlight the transmitter's knowledge of the available side-information at the receivers.

For the DD case, $S^{t-1}$ is added to the inputs of $f_t(.)$, while under DN model, only of the channels is revealed to the transmitter up to time $(t-1)$. Rather than enumerating all possibilities, we present an example, to clarify the encoding constraints. Suppose the transmitter is knows the side-information available to $\msf{Rx}_2$ (semi-blind), and has access to the delayed CSI from $\msf{Rx}_1$ (DN model), then, we have
\begin{align}
\label{eq_enc_function_example}
X[t] = f_t\lp W_{1|2}, \bar{W}_{1|2}, W_2, S_1^{t-1}, \msf{PI} \rp,
\end{align}

\noindent \underline{\bf Decoding:} Each receiver $\msf{Rx}_i$, $i=1,2$ knows its own CSI across entire transmission block $S^n_i$, and the CSI $S^n_{\bar{i}}$ if the other receiver $\bar{i}$ provides feedback. Under scenario NN it uses a decoding function $\varphi_{i,n}\left( Y_i^n, S_i^n, W_{\bar{i}|i} \right)$ to get an estimate $\widehat{W}_i$ of $W_i$, while under scenario DD the decoding function becomes $\varphi_{i,n}\left( Y_i^n, S^n, W_{\bar{i}|i} \right)$ where $S^n=(S_1^n, S^n_2)$. Note that under scenario DN, only the no-feedback receiver has global $S^n$. An error occurs whenever $\widehat{W}_i \neq W_i$. The average probability of error is given by
\begin{align}
\lambda_{i,n} = \mathbb{E}[P(\widehat{W}_i \neq W_i)],
\end{align}
where the expectation is taken with respect to the random choice of the transmitted messages.

\noindent \underline{\bf Capacity region:} We say that a rate pair $(R_1,R_2)$ is achievable, if there exists a block encoder at the transmitter, and a block decoder at each receiver, such that $\lambda_{i,n}$ goes to zero as the block length $n$ goes to infinity. The capacity region, $\mathcal{C}$, is the closure of the set of the achievable rate pairs. Throughout the paper, we will distinguish the capacity region under different assumptions. For example, $\mathcal{C}^{\mathrm{blind}}_{\mathrm{NN}}$ is the capacity region of the two-user broadcast erasure channels with a blind transmitter and no CSIT.

% , and is denoted by $\mathcal{C}$ for the non-blind scenario and by $\mathcal{C}^\mathrm{blind}$ for the blind transmitter assumption.

%%%%%%%%%%%%%%%%%%%%%%%%%%%%%%%%%%%%%%%%%%%%%%%%%%%

\section{Main Results}
\label{Section:Main_BIC}

In this section, we present the main contributions of this paper and provide some insights and intuitions about the findings.

\subsection{Statement of the Main Results}

We start with scenarios in which we characterize the capacity region, and then, we present cases for which we derive new inner-bounds. In Theorem~\ref{THM:Capacity_Out_BIC_No}, for the no CSIT scenario, we establish the capacity region with a non-blind transmitter. We will highlight the importance of side-information at the weaker receiver and how a semi-blind transmitter may achieve the same region in Remarks~\ref{remark:weakvsstrong} and~\ref{remark:semiblind}, respectively. Next, we present new capacity results when (some) delayed CSI is available to the transmitter in Theorem~\ref{THM:Capacity_Out_BIC_Delayed}. Next, for the no CSIT assumption and a blind transmitter, in Theorem~\ref{THM:Blind} we present new conditions beyond~\cite{kao2016blind} under which the capacity region is achievable, and a new achievable region is presented in Theorem~\ref{THM:Blind-Ach}.

% Then, we focus on the blind transmitter, still under the no CSIT assumption, and in Theorems~\ref{THM:Blind} and~\ref{THM:Blind-Ach} establish a new achievable region and a set of conditions under which the outer-bound region of Theorem~\ref{THM:Capacity_Out_BIC_No} can still be achieved. Finally, Theorem~\ref{THM:Capacity_Out_BIC_Delayed} focuses on the delayed CSIT setting and presents our new findings for various assumption summarized in Table~\ref{Table:Summary}.

\begin{theorem}
\label{THM:Capacity_Out_BIC_No}
For the two-user broadcast erasure channel with a non-blind transmitter and no CSIT as described in Section~\ref{Section:Problem_BIC}, we have
\begin{equation}
\label{Eq:Capacity_Out_BIC_No}
\mathcal{C}^\mathrm{non-blind}_\mathrm{NN} \equiv
\left\{ \begin{array}{ll}
0 \leq \beta^{\mathrm{no}}_1 R_1 + R_2 \leq \left( 1 - \delta_2 \right), & \\
0 \leq  R_1 + \beta^{\mathrm{no}}_2 R_2 \leq \left( 1 - \delta_1 \right). &
\end{array} \right.
\end{equation}
where
\begin{align}
\label{Eq:Beta_No}
\beta^{\mathrm{no}}_i = \epsilon_{\bar{i}} \min \left\{ \frac{1-\delta_{\bar{i}}}{1-\delta_i}, 1 \right\}.
\end{align}
\end{theorem}

The derivation of the outer-bounds has two main ingredients. First, as detailed in upcoming Remark \ref{remark:degraded}, even
the channel is not degraded, the stronger receiver can decode both messages regardless of the values of $\epsilon_1$ and $\epsilon_2$ for receiver cache. Second, as detailed in upcoming Lemma \ref{Lemma:Leakage_BIC_No}, we derive an extremal entropy inequality between the two receivers that captures the availability of receiver-end side-information, including the channel state and cache index information. The outer-bound region holds for the non-blind setting and thus, includes the capacity region with a blind transmitter as well. The following two remarks provide further insights.

\begin{remark}[Simplified expressions and degradedness]
\label{remark:degraded}
Without loss of generality, assume $\delta_2 \geq \delta_1$, meaning that receiver $1$ has a stronger channel. Then, the region of Theorem~\ref{THM:Capacity_Out_BIC_No} can be written as
\begin{equation}
\label{Eq:Capacity_Out_BIC_No_Simplified}
\left\{ \begin{array}{ll}
0 \leq \epsilon_2 \frac{1-\delta_2}{1-\delta_1} R_1 + R_2 \leq \left( 1 - \delta_2 \right), & \\
0 \leq R_1 + \epsilon_1 R_2 \leq \left( 1 - \delta_1 \right).
\end{array} \right.
\end{equation}
Unlike the scenario with no side-information at the receivers, this assumption does not mean the channel is degraded. However, the stronger receiver, ${\sf Rx}_1$ in this case, will be able to decode both $W_1$ and $W_2$ by  end of the communication block regardless of the values of $\epsilon_1$ and $\epsilon_2$. The reason is as follows. After decoding $W_1$, receiver ${\sf Rx}_1$ has access to the side-information of receiver ${\sf Rx}_2$, \emph{i.e.} $W_{1|2}$, and can emulate the channel of ${\sf Rx}_2$ as it has a stronger channel ($\delta_2 \geq \delta_1$). Finally, we note that although the stronger receiver is able to decode both messages, this does not imply that the stronger receiver will have a higher rate. As an example, suppose  $\delta_1 = 1/3, \delta_2 = 1/2, \epsilon_1 = 2/3,$ and $\epsilon_2 = 1/6$. Then, from the Theorem~\ref{THM:Capacity_Out_BIC_No}, the maximum sum-rate point is:
\begin{align}
\label{Eq:weakerhashigherrate}
\left( R_1, R_2 \right) = \left( \frac{4}{11}, \frac{5}{11} \right).
\end{align}
\end{remark}

\begin{remark}[Importance of side-information at the weaker receiver]
\label{remark:weakvsstrong}
Under the same assumption of the previous remark, $\delta_2 \geq \delta_1$, from the outer-bounds of Theorem~\ref{THM:Capacity_Out_BIC_No}, we conclude that if the weaker receiver has no side-information, \emph{i.e.} $\epsilon_2 = 1$, then, the capacity region is the same as having no side-information at either receivers. In other words, as long as the weaker receiver has no side-information, additional information at the stronger receiver does not enlarge the region. On other other hand, any side-information at the weaker receiver results in an outer-bound region that is strictly larger than the capacity region with no side-information at either receivers.
\end{remark}

Based on these remarks, we can provide more details on the achievability protocol. Under the no-CSIT assumption, the stronger receiver will eventually be able to decode both messages. Thus, the first step is to deliver the message intended for the weaker receiver. The stronger receiver will be able to decode this message faster than the intended receiver and thus, in the second step, we include the part of the message for the stronger user that is available at the weaker receiver. In other words, this second step is beneficial to both receivers. We note that to accomplish this task, the transmitter at least needs to know the side-information of the weaker receiver. This latter fact is further explained in the Corollary below. During the final step, the remaining part of the message intended for the stronger receiver is delivered.

As will be detailed in Section~\ref{Section:Achievability_NonBlind}, to achieve the outer-bounds, indeed the transmitter only needs to know the side-information available to the weaker receiver. Thus, we have

\begin{corollary} \label{remark:semiblind}
The capacity region with a semi-blind transmitter $\mathcal{C}^{\mathrm{semi-blind}}_{\mathrm{NN}}$ equals to $\mathcal{C}^{\mathrm{non-blind}}_{\mathrm{NN}}$ in Theorem \ref{THM:Capacity_Out_BIC_No} if the cache index information \eqref{eq_cacheModel} of the weaker receiver (link with larger erasure probability) is known at the transmitter.
\end{corollary}

%{\color{red}
The following lemma from~\cite{ghorbel2016content} establishes the outer-bounds on the capacity region of the two-user broadcast erasure channel with a non-blind transmitter and delayed CSIT, $\mathcal{C}^{\mathrm{non-blind}}_{\mathrm{DD}}$. We provide several new achievability strategies to achieve these bounds when the transmitter has less knowledge compared to what these bounds assume, and establish interesting capacity results.

\begin{lemma}[\cite{ghorbel2016content}]
\label{LEM:Capacity_Out_BIC_Delayed}
For the two-user broadcast erasure channel with a non-blind transmitter and delayed CSIT as described in Section~\ref{Section:Problem_BIC}, we have outer-bound region
\begin{equation}
\label{Eq:Capacity_Out_BIC_Delayed}
\mathcal{C}^{\mathrm{non-blind}}_{\mathrm{DD}} \subseteq
\left\{ \begin{array}{ll}
\vspace{1mm} 0 \leq R_i \leq (1-\delta_i), &  i = 1,2, \\
\beta^{\mathrm{delayed}}_i R_i + R_{\bar{i}} \leq \left( 1 - \delta_{\bar{i}} \right), & i = 1,2.
\end{array} \right.
\end{equation}
where
\begin{align}
\label{Eq:Beta_Delayed}
\beta^{\mathrm{delayed}}_i = \epsilon_{\bar{i}}\frac{1-\delta_{\bar{i}}}{1-\delta_{1}\delta_{2}}.
\end{align}
\end{lemma}

\noindent
Now, we show that this outer-bound region is achievable under the following scenarios.

\begin{theorem}
\label{THM:Capacity_Out_BIC_Delayed}
For the two-user broadcast erasure channel, the capacity region is achieved when:\\
\noindent {\bf Case A}: with a blind transmitter and global delayed CSIT, the capacity region $\mathcal{C}^{\mathrm{blind}}_{\mathrm{DD}}$ equals to \eqref{Eq:Capacity_Out_BIC_Delayed} when the channel is symmetric (\emph{i.e.} $\delta_1 = \delta_2 = \delta$ and $\epsilon_1 = \epsilon_2 = \epsilon$);

\noindent {\bf Case B}: with the transmitter knowing full side-information from one receiver and
only the delayed CSI of the other receiver (e.g., the semi-blind-transmitter case with $\epsilon_1 = 0$ for $\msf{Rx}_1$), the capacity region $\mathcal{C}^{\mathrm{semi-blind}}_{\mathrm{DN}}$ (and thus $\mathcal{C}^{\mathrm{semi-blind}}_{\mathrm{DD}}$) equals to \eqref{Eq:Capacity_Out_BIC_Delayed}

\noindent {\bf Case C}: with a non-blind transmitter having access to only the delayed CSI of one receiver, the capacity region $\mathcal{C}^{\mathrm{non-blind}}_{\mathrm{DN}}$ equals to \eqref{Eq:Capacity_Out_BIC_Delayed}.

% The transmitter has access to the delayed CSI of at least one receiver (\emph{i.e.} semi-blind or non-blind transmitter), and at least one receiver has full side-information (\emph{e.g.}, having $\epsilon_1 = 0$ for $\msf{Rx}_1$).
\end{theorem}
\noindent Note that having full side-information at a receiver (as in Case B above) immediately implies the transmitter is not blind with respect to that receiver. For the blind-transmitter case, if $\epsilon_i = 0$ then the transmitter knows $\msf{Rx}_i$ has full side-information, and thus $\mathcal{C}^{\mathrm{blind}}_{\mathrm{DN}}$ is also partially known from Case B.

Without the global channel state and/or cache index information from both receivers, our new achievability results of Theorem~\ref{THM:Capacity_Out_BIC_Delayed} differ significantly from the those in~\cite{ghorbel2016content}. In particular in~\cite{ghorbel2016content}, overheard bits and cached bits are both known at the transmitter to create network coding opportunities simultaneously benefit for both receivers. In our achievability, network coding opportunities can only be opportunistically created. Rather interestingly, in three cases identified by our Theorem, transmitter blindness or one-sided feedback may not result in any capacity loss compared with \cite{ghorbel2016content}. %The outer-bound region of Lemma~\ref{LEM:Capacity_Out_BIC_Delayed} can be obtained from the %results of~\cite{ghorbel2016content}.
%However, we  provide an alternative proof for Lemma~\ref{LEM:Capacity_Out_BIC_Delayed} in %Appendix~\ref{Section:DDBounds}.

%}

To achieve capacity $\mathcal{C}^{\mathrm{blind}}_{\mathrm{DD}}$ in Case A, we present an opportunistic protocol where the transmitter first sends out linear combinations of the packets for both receivers. Then, using the feedback signals and the fact that some of the packets for one receiver are available at the other, the transmitter sends bits for intended for one receiver in such as way to help one receiver remove interference and the other to obtain new information about its bits. Depending on channel parameters, this may follow with a multicast phase. This idea could be interpreted as an opportunistic reverse network coding for erasure channels. Both Cases B and C focus on cached capacity with only single-user delayed CSI.  Interestingly, our four-phase opportunistic network coding for Case C will also create blind side-information at the ``N" user for the recycled bits. Indeed, the last two phases in Case C is modified form the achievability for Case B. Compared with \cite{lin2019no}, the new ingredient  for Case C is the non-blind cache, and we propose to generalize \cite{lin2019no}
by carefully mixing the fresh cached bits with the recycled un-cached bits. Specifically,  the recycling in \cite{lin2019no} is done by mixing two pre-encoded bit sequences. On the contrary to \cite{lin2019no} where the input sequence of each pre-encoder is always recycled, it may contain fresh bits in our Case C as detailed in Sec. \ref{Section:Section:Achievability_BIC_DN}.

%The optimality proof the theorem relies on capturing the impact of the available receiver side-information, channel erasure, and delayed CSI, on the information leakage from one receiver to the other. These outer-bounds do not rely on the transmitter being blind, thus, will hold even if the transmitter knows exactly what each receiver has as side-information. In essence, rather interestingly, transmitter blindness does not result in any capacity loss.
%The only other effort in capturing such conditions minus CSI feedback is the results of~\cite{kao2016blind}, but those are not tight. Thus, to the best of our knowledge, the techniques presented here are the first tight outer-bounds. We present the bounds for any possible encoding function, but our achievability is based on linear operations.

The following two theorems focus on the no-CSIT blind-transmitter assumption. The first one identifies conditions under which the outer-bound region of Theorem~\ref{THM:Capacity_Out_BIC_No} can be achieved even when the transmitter does not know what side-information is available to each receiver. The second theorem presents an achievable region when the stronger receiver has full side-information, but this achievable region does not match the outer-bounds.

\begin{theorem}
\label{THM:Blind}
For the two-user broadcast erasure channel with no channel feedback, a {blind transmitter}, and available receiver side-information as described in Section~\ref{Section:Problem_BIC}, the capacity $\mathcal{C}^\mathrm{blind}_\mathrm{NN}$ equals to $\mathcal{C}^\mathrm{non-blind}_\mathrm{NN}$ in Theorem~\ref{THM:Capacity_Out_BIC_No} when:
\begin{enumerate}

\item When $\delta_2 \geq \delta_1$ and $\epsilon_2 \in \{ 0, 1 \}$;

\item Symmetric setting: when $\delta_1 = \delta_2 = \delta$, and $\epsilon_1 = \epsilon_2 = \epsilon$.

\end{enumerate}
\end{theorem}

\begin{theorem}
\label{THM:Blind-Ach}
For the two-user broadcast erasure channel with no channel feedback, a {blind transmitter}, as described in Section~\ref{Section:Problem_BIC}, when $\delta_2 \geq \delta_1$ and $\epsilon_1 = 0$, we can the following region is achievable:
\begin{equation}
\label{Eq:Capacity_Inner_BIC_No}
\mathcal{R}^\mathrm{in} \equiv
\left\{ \begin{array}{ll}
0 \leq  \left( \epsilon_2+\delta_1(1-\epsilon_2) \right) \frac{1-\delta_2}{1-\delta_1} R_1 + R_2 \leq \left( 1 - \delta_2 \right), & \\
0 \leq  R_1 \leq \left( 1 - \delta_1 \right). &
\end{array} \right.
\end{equation}
\end{theorem}

\noindent
We note that for $\delta_1 = 0$ (no erasure at the stronger receiver), the inner-bound region of Theorem~\ref{THM:Blind-Ach} matches the outer-bounds of Theorem~\ref{THM:Capacity_Out_BIC_No}, \emph{i.e.} $\mathcal{R}^\mathrm{in} \equiv \mathcal{C}^\mathrm{blind}_\mathrm{NN}$ .
% {\color{red} The next paragraph needs revision.}

Note that the blind-transmitter case under no-CSIT was also considered in~\cite{kao2016blind} but only the weaker receiver has side-information. In that setting, outer and inner bounds were presented that match only when erasure probabilities are equal to zero, which is no longer an erasure setting. In contrast, in this work, we have the capacity region $\mathcal{C}^\mathrm{non-blind}_\mathrm{NN}$ in Theorem~\ref{THM:Capacity_Out_BIC_No} for the non-blind-transmitter case, which recovers the outer-bounds of~\cite{kao2016blind} as a special case. The capacity region $\mathcal{C}^\mathrm{blind}_\mathrm{NN}$ of blind index coding over no-CSIT broadcast erasure channel remains open in general.

%%%%%%%%%%%%%%%%%%%%%%%%%%%%%%%%%%%%%%%%%%%%%%%%%%%%%%%%%%%%%%%

%%%%%%%%%%%%%%%%%%%%%%%%%%%%%%%%%%%%%%%%%%%%%%%%%%%%%%%%

\subsection{Illustration of the results}

In this subsection, we briefly illustrate the results through a few examples to further clarify and discuss some of the insights and intuitions provided above.

We start with Theorem~\ref{THM:Capacity_Out_BIC_Delayed} where the transmitter has access to (some) delayed CSI. Figure~\ref{Fig:Region-BIC-delayed} illustrates the capacity region $\mathcal{C}^{\mathrm{blind}}_{\mathrm{DD}}=\mathcal{C}^{\mathrm{non-blind}}_{\mathrm{DN}}$  from Theorem~\ref{THM:Capacity_Out_BIC_Delayed} when $\delta_1 = \delta_2 = \delta = 0.5$ and $\epsilon_1 = \epsilon_2 = \epsilon \in \{ 0, 0.5, 1\}$. In particular, $\epsilon = 1$ is the case in which no side-information is available and our results recover~\cite{GatzianasGeorgiadis_13}. The other extreme is $\epsilon = 0$ where the entire message of one user is available to the other and maximum individual point-to-point rates can be achieved. Finally, $\epsilon = 0.5$ is an intermediate case and the capacity is strictly larger than that of no side-information.

\begin{figure}[!ht]
\centering
\includegraphics[width = 0.5\columnwidth]{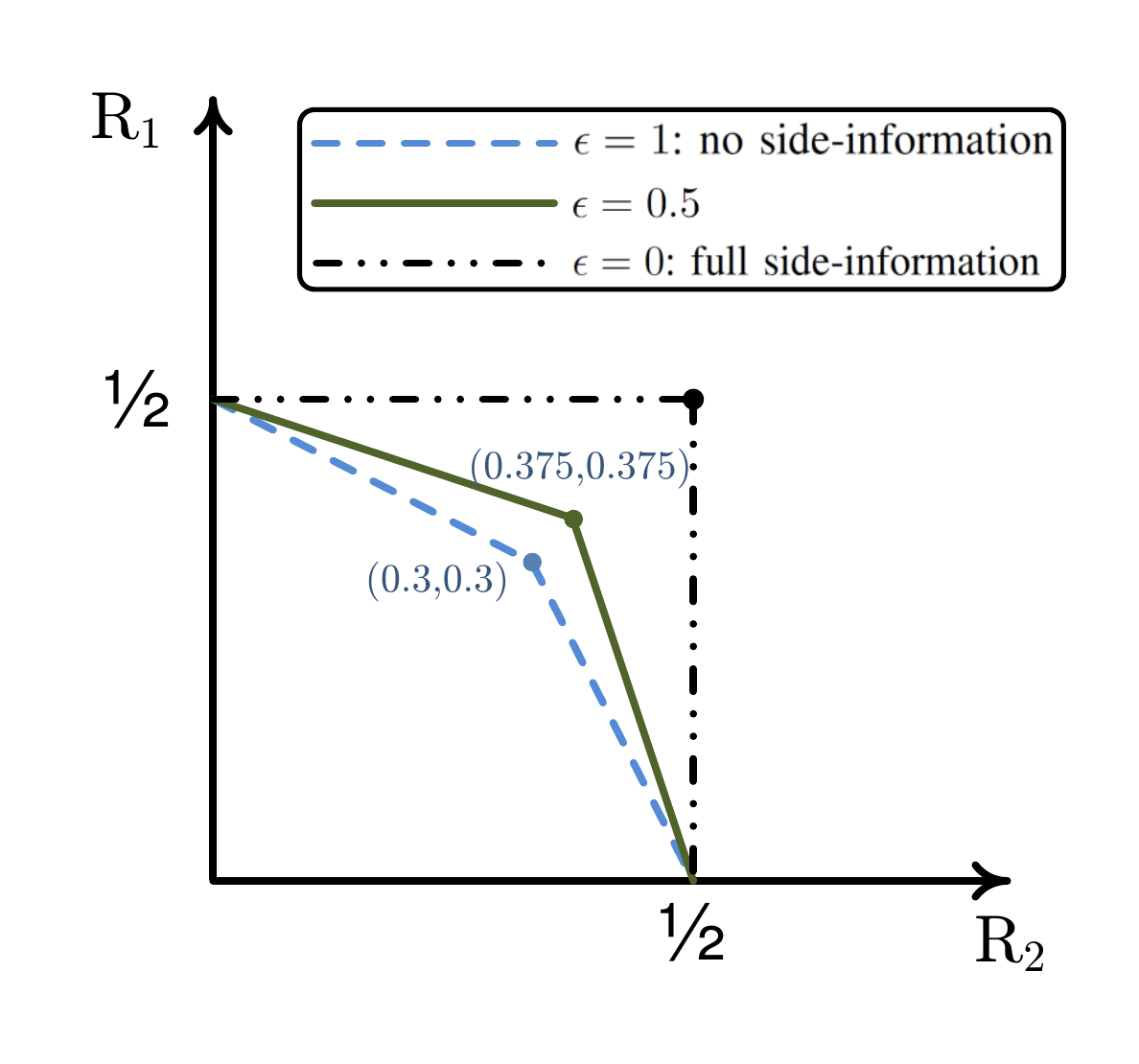}
\caption{Capacity region $\mathcal{C}^{\mathrm{blind}}_{\mathrm{DD}}=\mathcal{C}^{\mathrm{non-blind}}_{\mathrm{DN}}$  for $\delta_1 = \delta_2=\delta = 0.5$ and $\epsilon_1 = \epsilon_2=\epsilon \in \{ 0, 0.5, 1\}$.\label{Fig:Region-BIC-delayed}}
\end{figure}

We then consider the no-CSIT scenario. For the first example in this case, we consider $\delta_1 = \frac{1}{2}$ and $\delta_2 = \frac{3}{4}$. The capacity region of the broadcast erasure channel with these parameters and no side-information at either receiver is described by all non-negative rates satisfying:
\begin{align}
\label{Eq:CapacityNoSideInfo}
\frac{1}{2}R_1 + R_2 \leq \frac{1}{4}.
\end{align}
Note that we no side-information, the channel is degraded. Further, as discussed earlier, as long as the weaker receiver (${\sf Rx}_2$ in this case) has no side-information, \emph{i.e.} $\epsilon_2 = 1$, the capacity region remains identical to the one described in \eqref{Eq:CapacityNoSideInfo} with no side-information at either receivers. This region is included in Figure~\ref{Fig:Example1}(a) and (b) as benchmark. Note that significant caching or index coding gains are obtained for all settings presented in Figure~\ref{Fig:Example1}(a) and (b).

\begin{figure}[!ht]
\centering
\subfigure[]{\includegraphics[width = 0.45\columnwidth]{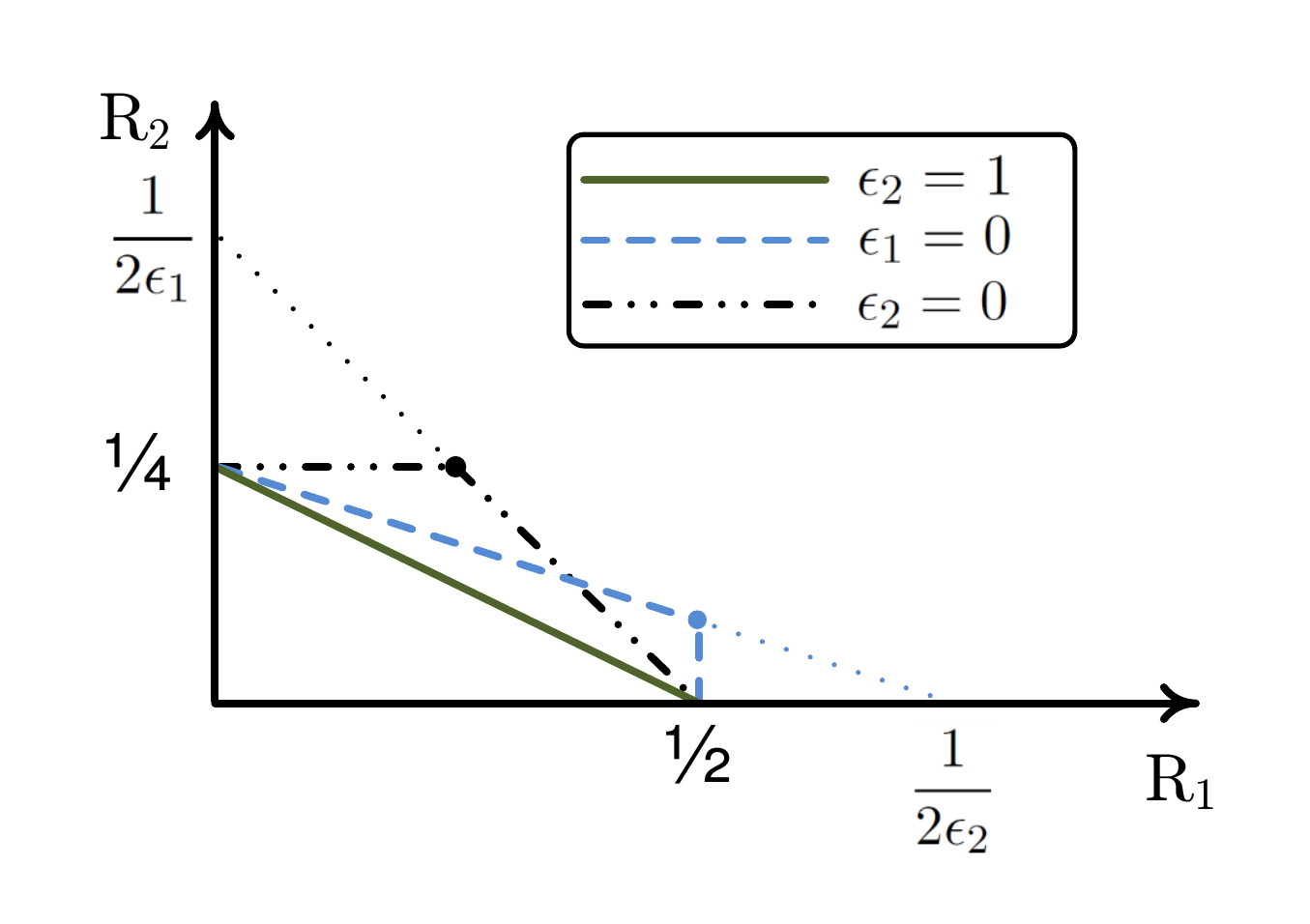}}
\subfigure[]{\includegraphics[width = 0.45\columnwidth]{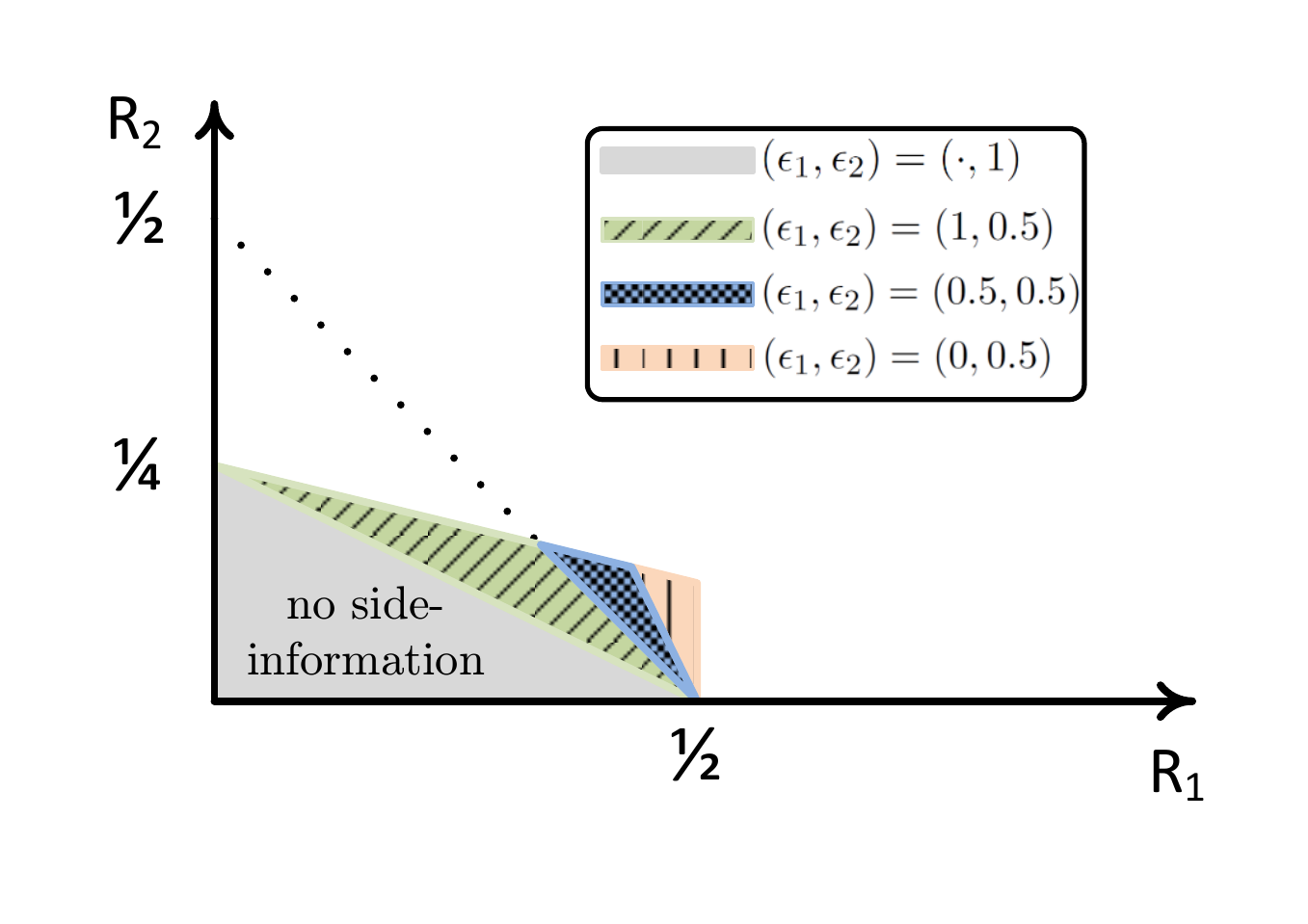}}
\caption{(a) Illustration of the capacity region $\mathcal{C}^{\mathrm{non-blind}}_{\mathrm{NN}}$ when ${\sf Rx}_i$ has full side-information ($\epsilon_i = 0$); (b) increase in the achievable rates as $\epsilon_1$ goes from $1$ to $0$ for $\epsilon_2 = 0.5$.\label{Fig:Example1}}
\end{figure}

Next, we examine the region $\mathcal{C}^{\mathrm{non-blind}}_{\mathrm{NN}}$ when one receiver has full side-information. Figure~\ref{Fig:Example1}(a) includes both these cases and depicts how the capacity region enlarges as more side-information is available to the receivers. We note that with full side-information at both receivers ($\epsilon_1 = \epsilon_2 = 0$), maximum individual rates given by $R_i = (1 - \delta_i)$, $i=1,2$, are achievable simultaneously. Figure~\ref{Fig:Example1}(b) depicts the gradual increase in achievable rates when $\epsilon_2 = \frac{1}{2}$ and $\epsilon_1$ goes from $1$ (no side information) to $0$ (full side-information).
Note that receiver ${\sf Rx}_1$ has a stronger channel so the illustrated region also equals to the capacity region $\mathcal{C}^{\mathrm{semi-blind}}_{\mathrm{NN}}$ when the transmitter has no cache index information of ${\sf Rx}_1$.

\begin{figure}[!ht]
\centering
\subfigure[]{\includegraphics[width = 0.45\columnwidth]{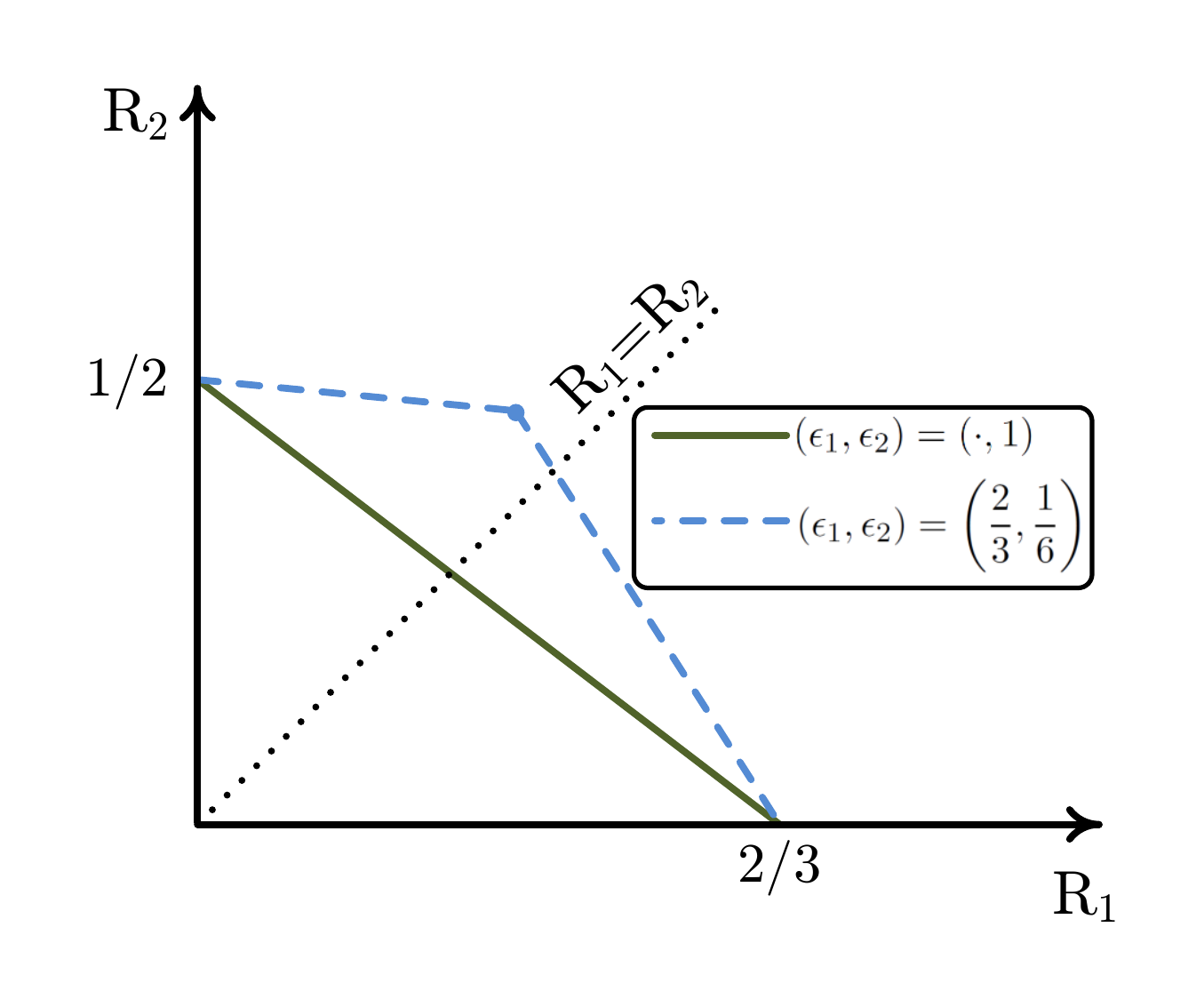}}
\subfigure[]{\includegraphics[width = 0.45\columnwidth]{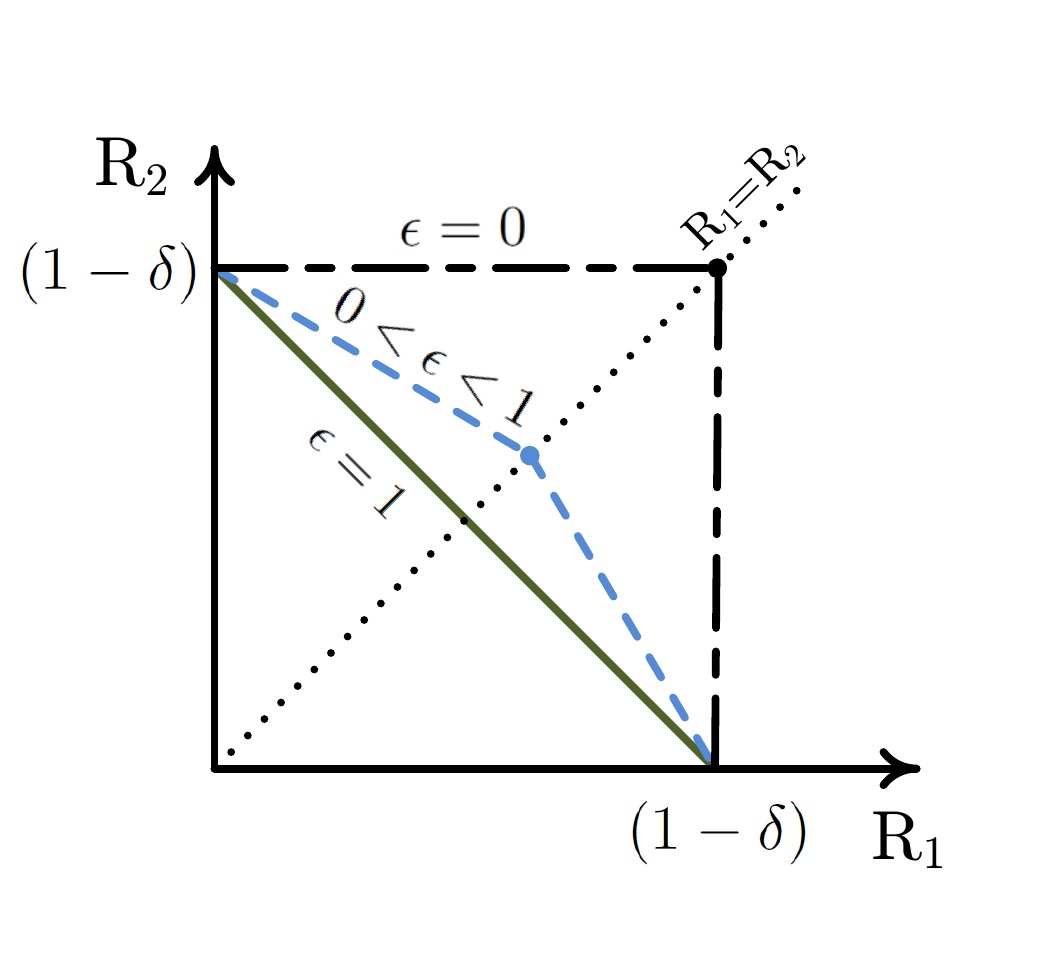}}
\caption{(a) Although the stronger receiver will always be able to decode both messages, the weaker receiver may have a higher rate. In this example for $\mathcal{C}^{\mathrm{non-blind}}_{\mathrm{NN}}$, we have $\delta_1 = \frac{1}{3}, \delta_2 = \frac{1}{2}, \epsilon_1 = \frac{2}{3},$ and $\epsilon_2 = \frac{1}{6}$; (b) Capacity region with symmetric parameters $\delta_1 = \delta_2 = \delta$ and $\epsilon_1 = \epsilon_2 = \epsilon$. The maximum sum-rate point is given by $\frac{2(1-\delta)}{(1+\epsilon)}$. The capacity region with symmetric parameters can be achieved under the blind-transmitter scenario as well. See Theorem \ref{THM:Blind}.\label{Fig:Example2}}
\end{figure}

As the second example, we consider $\delta_1 = \frac{1}{3}$ and $\delta_2 = \frac{1}{2}$. Similar to the previous case for capacity region $\mathcal{C}^{\mathrm{non-blind}}_{\mathrm{NN}}$, receiver ${\sf Rx}_1$ has a stronger channel. However, as illustrated in Figure~\ref{Fig:Example2}(a), with $\epsilon_1 = \frac{2}{3}$ and $\epsilon_2 = \frac{1}{6}$, the weaker receiver has a higher rate as given in \eqref{Eq:weakerhashigherrate}. Finally, Figure~\ref{Fig:Example2}(b) depicts the capacity region $\mathcal{C}^{\mathrm{non-blind}}_{\mathrm{NN}}=\mathcal{C}^{\mathrm{blind}}_{\mathrm{NN}}$ with symmetric channel parameters ($\delta_1 = \delta_2 = \delta$ and $\epsilon_1 = \epsilon_2 = \epsilon$): with no side-information, the maximum achievable sum-rate is $(1-\delta)$; and with side-information (even blind), the maximum sum-rate is given by:
\begin{align}
\frac{2(1-\delta)}{(1+\epsilon)}.
\end{align}

\begin{figure}[!ht]
\centering
\includegraphics[width = 0.55\columnwidth]{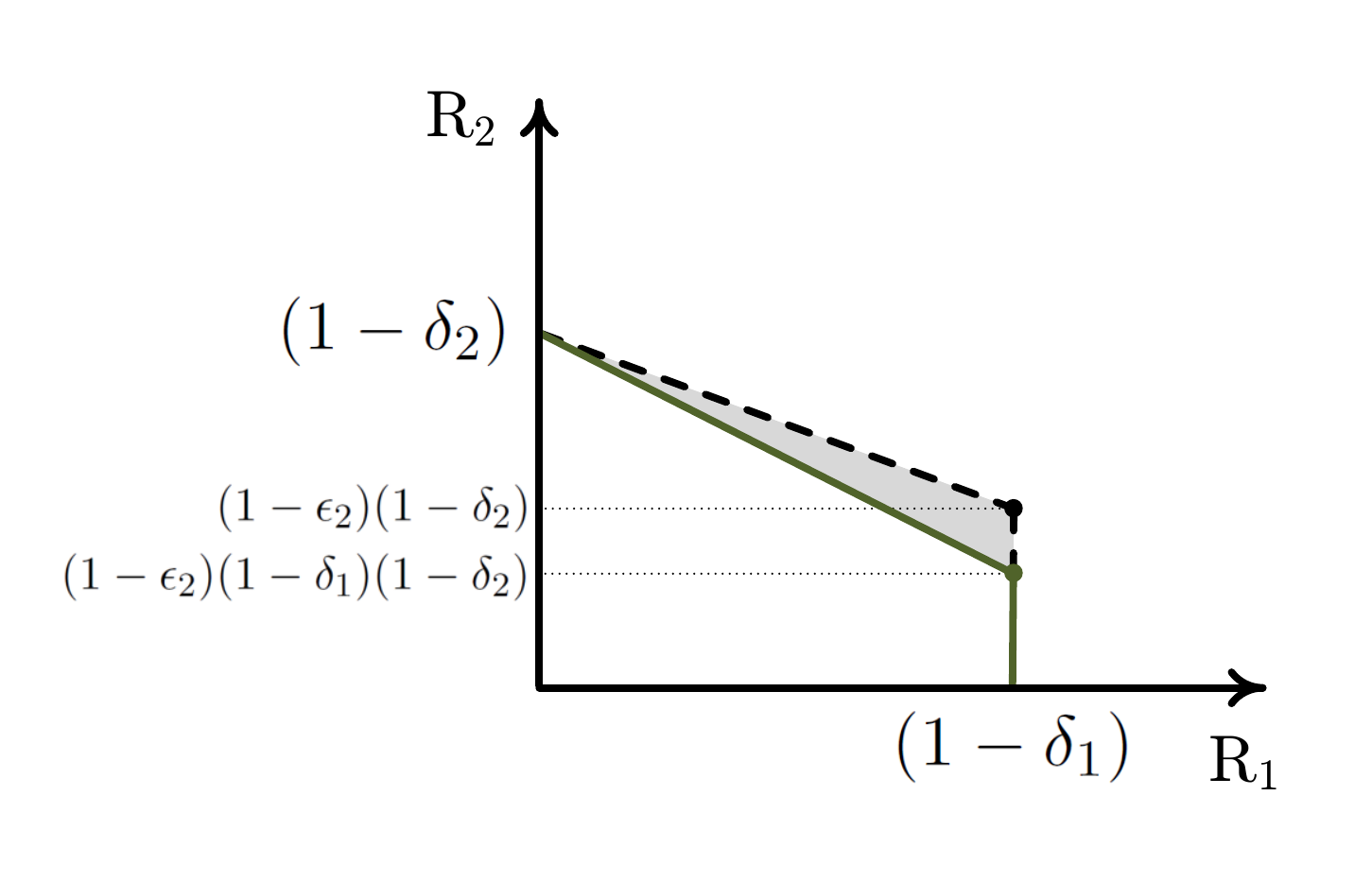}
\caption{Comparing the inner-bounds of Theorem~\ref{THM:Blind-Ach} for $\mathcal{C}^{\mathrm{blind}}_{\mathrm{NN}}$ to the outer-bounds of Theorem~\ref{THM:Capacity_Out_BIC_No} for $\delta_2 \geq \delta_1$ and $\epsilon_1 = 0$. Black dashed lines define the outer-bound region, while the solid green lines are the inner-bounds. The shaded area is the gap between the two.\label{Fig:StrongerwithFullSideInfo}}
\end{figure}

Finally, Figure~\ref{Fig:StrongerwithFullSideInfo} illustrates the outer-bounds of Theorem~\ref{THM:Capacity_Out_BIC_No} and the inner-bounds of Theorem~\ref{THM:Blind-Ach} for $\delta_2 \geq \delta_1$ and $\epsilon_1 = 0$. In other words, in this case, the stronger receiver has full side-information. We note that the inner-bounds of Theorem~\ref{THM:Blind-Ach} for the blind-transmitter case match the outer-bounds when $\delta_1 = 0$.

\subsection{Organization of the Proofs}

In the following sections, we provide the proof of our main contributions. We prove the capacity region of Theorem~\ref{THM:Capacity_Out_BIC_No} in Sections~\ref{Section:Converse_BIC} and~\ref{Section:Achievability_NonBlind}. We then move to the achievability part of Theorem~\ref{THM:Capacity_Out_BIC_Delayed} as they are capacity-achieving and include several interesting new ingredients, as in Section \ref{Section:Achievability_BIC_DDSymmetric} and \ref{Section:Section:Achievability_BIC_DN}. The proofs of other results are deferred to the Appendix.

%%%%%%%%%%%%%%%%%%%%%%%%%%%%%%%%%%%%%%%%%%%%%%%%%%%%
\section{Converse Proof of $\mathcal{C}^\mathrm{non-blind}_\mathrm{NN}$ in Theorem~\ref{THM:Capacity_Out_BIC_No}}
\label{Section:Converse_BIC}

In this section, we derive the outer-bounds of Theorem~\ref{THM:Capacity_Out_BIC_No}. We note that as the capacity region of the non-blind setting includes that of the blind assumption, the derivation in this section is for the non-blind transmitter case and the bounds apply to the blind-transmitter case as well.

The point-to-point outer-bounds, \emph{i.e.} $R_i \leq (1-\delta_i)$, are those of erasure channels, and thus, omitted. Without loss of generality, we assume $\delta_2 \geq \delta_1$, meaning that receiver $1$ has a stronger channel. As discussed in Remark~\ref{remark:degraded}, unlike the scenario with no side-information at the receivers, this assumption does not mean the channel is degraded. In what follows, we derive the following outer-bounds:
\begin{align}
\label{Eq:ConverseBoundN}
& \mathbf{B1}: \frac{\epsilon_2(1-\delta_2)}{(1-\delta_1)} R_1 + R_{2} \leq \left( 1 - \delta_2 \right), \nonumber \\
& \mathbf{B2}: R_1 + \epsilon_1 R_2 \leq \left(1-\delta_1 \right).
\end{align}

Suppose rate-tupe $\lp R_1, R_2 \rp$ is achievable. We first derive $\mathbf{B2}$ to get some insights.

\noindent {\bf Derivation of $\mathbf{B2}$}: As discussed in Remark~\ref{remark:degraded}, the stronger receiver, ${\sf Rx}_1$ in this case, is able to decode both messages by the end of the communication block using its available side information. Thus, we have
\begin{align}
\label{Eq:StrongerRxDecodes}
&H(W_1|G^n) + H(\bar{W}_{2|1}|G^n) \leq \notag \\ &I(W_1, \bar{W}_{2|1} ;Y_1^n|W_{2|1},G^n) + n \upxi_n \leq H(Y_1^n|W_{2|1},G^n) + n \upxi_n \leq n (1-\delta_1) + n \upxi_n.
\end{align}
We also note that
\begin{align}
\label{Eq:FractionRate2}
n H(\bar{W}_{2|1}|G^n)= \sum^{m_2}_{\ell=1} H\big((1-E_1[\ell])W_2[\ell]\big|E_1[\ell]\big)=n \epsilon_1 H(W_2)=n \epsilon_1 R_2.
\end{align}
Thus, from \eqref{Eq:StrongerRxDecodes} and \eqref{Eq:FractionRate2}, we get
\begin{align}
n (R_1 + \epsilon_1 R_2) \leq n (1-\delta_1) + n \upxi_n.
\end{align}
Dividing both sides by $n$ and let $n \rightarrow \infty$, we get the second outer-bound in \eqref{Eq:ConverseBoundN}.

\noindent {\bf Derivation of $\mathbf{B1}$}: We enhance receiver $\msf{Rx}_1$ by providing the entire $W_2$ to it, as opposed to the already available $W_{2|1}$, and we note that this cannot reduce the rates. Moreover, motivated from the Derivation of $\mathbf{B2}$, this enhancement should only have limited rate increase. From the decentralized placement model \eqref{eq_cacheModel}, we define the global channel state and cache index information as
\begin{equation} \label{eq_statecahce}
G^n:=\{S^n_1, S^n_2, E^{nR_2}_1, E^{nR_1}_2\}:=\{S^n, E^n\}
\end{equation}
then, using
\begin{align}
\beta_1^{\mathrm{no}} = \frac{\epsilon_2(1-\delta_2)}{(1-\delta_1)},
\end{align}
we have
\begin{align}
n &\left( \beta_1^{\mathrm{no}} R_1 + R_2 \right) = \beta_1^{\mathrm{no}} H(W_1) + H(W_2) \nonumber \\
& \overset{(a)}= \beta_1^{\mathrm{no}} H(W_1|W_2, G^n) + H(W_2|W_{1|2}, G^n) \nonumber \\
& \overset{(\mathrm{Fano})}\leq \beta_1^{\mathrm{no}} I(W_1;Y_1^n|W_2, G^n) + I(W_2;Y_2^n|W_{1|2}, G^n) + n \upxi_n \nonumber \\
& = \beta_1^{\mathrm{no}} H(Y_1^n|W_2, G^n) - \beta_1^{\mathrm{no}} \underbrace{H(Y_1^n|W_1,W_2,G^n)}_{=~0} \nonumber \\
& \quad + H(Y_2^n|W_{1|2},G^n) - H(Y_2^n|W_{1|2},W_2,G^n) + n \upxi_n \nonumber \\
& \overset{(b)}\leq H(Y_2^n|W_{1|2},G^n) + 2n \upxi_n \nonumber \\
& \overset{(c)}\leq n \left( 1 - \delta_2 \right) + 2\upxi_n,
\end{align}
where $\upxi_n \rightarrow 0$ as $n \rightarrow \infty$; $(a)$ follows from the independence of messages and captures the enhancement of receiver $\msf{Rx}_1$; $(b)$ follows from Lemma~\ref{Lemma:Leakage_BIC_No} below; $(c)$ is true since the entropy of a binary random variable is at most one and the channel to the second receiver is only on a fraction $\lp 1 - \delta_2 \rp$ of the communication time. Dividing both sides by $n$ and let $n \rightarrow \infty$, we get the first outer-bound in \eqref{Eq:ConverseBoundN}.

\begin{lemma}
\label{Lemma:Leakage_BIC_No}
For the two-user broadcast erasure channel with no channel feedback but with non-blind side information at the receivers as described in Section~\ref{Section:Problem_BIC}, whether blind or not, and $\beta_1^{\mathrm{no}}$ given in (\ref{Eq:Beta_No}), we have
\begin{align}
H\left( Y_2^n | W_{1|2}, W_2,G^n \right) + n \upxi_n \geq \beta_1^{\mathrm{no}} H\left( Y_1^n | W_2,G^n \right),
\end{align}
where $\upxi_n \rightarrow 0$ as $n \rightarrow \infty$, where $G^n$ is the global channel state and cache index information in \eqref{eq_statecahce}.
\end{lemma}

\begin{proof}

We first have the following fact
\begin{align}
\label{Eq:withSideInfoN}
H\left( Y_2^n | W_{1|2}, W_2, G^n \right)  \geq \frac{1-\delta_2}{1-\delta_1} H\left( Y_1^n | W_{1|2}, W_2, G^n \right),
\end{align}
which is modified from our previous result \cite{MaddahAliTIT15}. For completeness, we still present the details in Appendix \ref{AppEqwithSideInfoN}. Now, to prove this Lemma with \eqref{Eq:withSideInfoN}, we note that
\begin{equation}\label{Eq:Deterministic}
0=H\left( Y_1^n | W_1, W_2, G^n \right) = H\left( Y_1^n | W_{1|2}, \bar{W}_{1|2}, W_2, G^n \right),
\end{equation}
where $\bar{W}_{1|2}$ is the complement of $W_{1|2}$ in $W_1$, and then
\begin{align}
&H\left( Y_1^n | W_{1|2}, W_2, G^n \right) =I(\bar{W}_{1|2}; Y_1^n|W_{1|2}, W_2, G^n ) \notag \\
=& H(\bar{W}_{1|2}|W_{1|2}, W_2, G^n)-H(\bar{W}_{1|2}|Y_1^n, W_{1|2}, W_2, G^n ). \label{Eq:RemovingSideInfo0}
\end{align}
Since $H\left( W_1 | Y_1^n, W_2, G^n \right) \leq n \upxi_n,$ the second term in the RHS of \eqref{Eq:RemovingSideInfo0} will also be less than $n \upxi_n$ due to the $W_1=(W_{1|2}, \bar{W}_{1|2})$ and the chain rule. For the first term in the RHS, as \eqref{Eq:FractionRate2}, it equals to
$
H(\bar{W}_{1|2}| E_2^n) = \epsilon_2H(W_1).
$
Then we get
\begin{align}
\label{Eq:RemovingSideInfoN}
H&\left( Y_1^n | W_{1|2}, W_2, G^n \right) \geq  \epsilon_2 H\lp W_1 \rp - n \upxi_n \overset{(a)} \geq \epsilon_2 H\lp Y_1^n | W_2, G^n \rp - n \upxi_n,
\end{align}
and $(a)$ is obtained from \eqref{Eq:Deterministic} as
\begin{align}
H\left( Y_1^n | W_1, W_2, G^n \right) = 0 \Rightarrow H\lp W_1 \rp \geq H\lp Y_1^n | W_2, G^n \rp. \nonumber
\end{align}
Finally, from \eqref{Eq:withSideInfoN} and \eqref{Eq:RemovingSideInfoN}, we obtain
\begin{align}
H&\left( Y_2^n | W_{1|2}, W_2, G^n \right) \overset{\eqref{Eq:withSideInfoN}}\geq \frac{1-\delta_2}{1-\delta_1} H\left( Y_1^n | W_{1|2}, W_2, G^n \right) \nonumber \\
& \overset{\eqref{Eq:RemovingSideInfoN}}\geq \frac{\epsilon_2\lp 1-\delta_2 \rp}{\lp 1-\delta_1 \rp} H\left( Y_1^n | W_2, G^n \right) - n \upxi_n \nonumber \\
& \overset{\eqref{Eq:Beta_No}}= \beta_1^{\mathrm{no}} H\left( Y_1^n | W_2, G^n \right) - n \upxi_n.
\end{align}
This completes the proof of Lemma~\ref{Lemma:Leakage_BIC_No}.
\end{proof}

%%%%%%%%%%%%%%%%%%%%%%%%%%%%%%%%%%%%%%%%%%%%%%%%%%%

%\section{Achievability Challenges for General Parameters}
%\label{Section:Challenges_BIC}

%%%%%%%%%%%%%%%%%%%%%%%%%%%%%%%%%%%%%%%%%%%%%%%%%%%

\section{Achievability Proof of $\mathcal{C}^\mathrm{non-blind}_\mathrm{NN}$ in Theorem~\ref{THM:Capacity_Out_BIC_No}}
\label{Section:Achievability_NonBlind}

%%%%%%%%%%%%%%%%%%%%%%%%%%%%%%%%%%%%%%%%%%%%%%%%%%%

In this section, we provide an achievability protocol for the non-blind transmitter case, and we show that the achievable regions matches the outer-bounds of Theorem~\ref{THM:Capacity_Out_BIC_No}. Thus, we characterize the capacity region of the problem when the transmitter is aware of the available side-information at the receivers. For the proof, without loss of generality, we assume $\delta_2 \geq \delta_1$, meaning that ${\sf Rx}_1$ has a stronger channel than ${\sf Rx}_2$.

Our achievability proof reveals a surprising result, only the cache index information $E^{nR_1}_2$ of the weaker receiver 2 in \eqref{eq_cacheModel} is needed at the transmitter. Thus the transmitter can be semi-blind as Corollary \ref{remark:semiblind} to achieve the outer-bounds. As a warm-up, we first focus on an example where $\epsilon_2 = 0$. In this case, receiver ${\sf Rx}_2$ (the weaker receiver) has full side-information of the message for ${\sf Rx}_1$, \emph{i.e.} $W_{1|2} = W_1$, and receiver ${\sf Rx}_1$ has access to $(1-\epsilon_1)$ of the bits intended for ${\sf Rx}_2$. The outer-bounds of Theorem~\ref{THM:Capacity_Out_BIC_No} in this case become:
\begin{equation}
\label{Eq:FullSideInfo_Weaker}
\left\{ \begin{array}{ll}
0 \leq R_2 \leq \left( 1 - \delta_2 \right), & \\
0 \leq R_1 + \epsilon_1 R_2 \leq \left( 1 - \delta_1 \right).
\end{array} \right.
\end{equation}
Thus, the non-trivial corner point is given by:
\begin{align}
\label{Eq:FullSideInfo_Weaker_Rates}
&R_1 = (1-\delta_1) - \epsilon_1 (1-\delta_2), \nonumber \\
&R_2 = (1-\delta_2).
\end{align}
In this case, we set
\begin{align}
\eta = \frac{R_1}{R_2} = \frac{(1-\delta_1)}{(1-\delta_2)} - \epsilon_1 > 0.
\end{align}

\noindent {\bf Achievability protocol for $\epsilon_2=0$ :} Recall $\eta=m_1/m_2$, we start with $m$ bits for ${\sf Rx}_2$ and $\eta m$ bits for ${\sf Rx}_1$. The achievability protocol is carried a single phase with two segments (another Phase will be added later for general $\epsilon_2$). The total communication time is set to $\frac{m}{(1-\delta_2)}.
$

\noindent {\bf Segment~a:} This segment has a total length of
\begin{align} \label{eq_NNta}
t_a = \frac{\epsilon_1m}{(1-\delta_1)} < \frac{m}{(1-\delta_2)},
\end{align}
where the last inequality is from $R_1>0$ in \eqref{Eq:FullSideInfo_Weaker_Rates}. During this segment, the transmitter sends $t_a$ of the random combinations of the $m$ bits intended for ${\sf Rx}_2$. During this segment, stronger ${\sf Rx}_1$ obtains
\begin{align}
(1-\delta_1)t_a = \epsilon_1m
\end{align}
random equations of the $m$ bits for ${\sf Rx}_2$ and in combination with the available side-information $W_{2|1}$ with $\left| W_{2|1} \right| = (1-\epsilon_1) m$, ${\sf Rx}_1$ has sufficient linearly independent equations to decode $W_2$ when the code length is large enough.

\noindent {\bf Segment~b:} This segment has a total length of
\begin{align} \label{eq_NNtb}
t_b = \frac{(1-\delta_1)-\epsilon_1(1-\delta_2)}{(1-\delta_1)(1-\delta_2)}m  = \frac{m}{(1-\delta_2)} - t_a > 0.
\end{align}
During this segment, the transmitter creates $t_b$ random linear combinations of the $\eta m$ bits in $W_1$, and creates the XOR of these combinations with $t_b$ random combinations of the $m$-bit $W_2$ for ${\sf Rx}_2$. The transmitter sends the resulting XORed sequence during Segment~b.

\noindent The decodability comes as follows. In segment b, ${\sf Rx}_1$ can remove the interference since $W_2$ is known from Segment a, and gets  $t_b(1-\delta_1)=\eta m$ linearly independent equations for decoding $W_1$ correctly. Also in this segment,
${\sf Rx}_2$ can remove the interference from $W_1$ using the side-information $W_{1|2} = W_1$, then the total linearly independent equations it has will be $(t_b+t_a)(1-\delta_2)=m$. Thus, by the end of the communication block, ${\sf Rx}_2$ can decode $W_2$, and ${\sf Rx}_1$ can decode both $W_1$ and $W_2$.

\noindent {\bf Achievable rates:} Since the total communication time is
\begin{align}
\frac{m}{(1-\delta_2)},
\end{align}
we immediately conclude the achievability of rates in \eqref{Eq:FullSideInfo_Weaker_Rates}.

Note that in the toy example aforementioned, only cache index information $E^{nR_1}_2$ for $W_{1|2}=W_1$ is used at the transmitter in Segment b, but not the other $E^{nR_2}_1 (W_{2|1})$.  Now we present the proof for the general case $\epsilon_2 \neq 0$ and show that the achievability also needs a ``semi-blind" transmitter. From \eqref{Eq:Capacity_Out_BIC_No_Simplified}, the non-trivial corner point of the region defined in \eqref{Eq:Capacity_Out_BIC_No} is given by:
\begin{align}
\label{Eq:GeneralRates}
&R_1 = \frac{(1-\delta_1)- \epsilon_1(1-\delta_2)}{1-\frac{\epsilon_1 \epsilon_2 (1-\delta_2)}{(1-\delta_1)}}, \nonumber \\
&R_2 = \frac{(1-\epsilon_2)(1-\delta_2)}{1-\frac{\epsilon_1 \epsilon_2 (1-\delta_2)}{(1-\delta_1)}}.
% &R_1 = (1-\delta_1) - \epsilon_1 R_2 = \frac{(1-\delta_1)- \epsilon_1(1-\delta_2)}{1-\frac{\epsilon_1 \epsilon_2 (1-\delta_2)}{(1-\delta_1)}}
\end{align}
We define
\begin{align}
\eta \overset{\triangle}= \frac{R_1}{R_2} = \frac{(1-\delta_1)-\epsilon_1(1-\delta_2)}{(1-\epsilon_2)(1-\delta_2)}.
\end{align}
We note that if
\begin{align}
(1-\epsilon_2+\epsilon_1) >  \frac{(1-\delta_1)}{(1-\delta_2)},
\end{align}
%$(1-\epsilon_2+\epsilon_1)$ is greater than $(1-\delta_1)/(1-\delta_2)$,
then $\eta < 1$.

\noindent {\bf Achievability protocol for general $\epsilon_2$:} We start with $m$ bits for ${\sf Rx}_2$ and $\eta m$ bits for ${\sf Rx}_1$. The achievability protocol is carried over two phases with the first phase having two segments as those for $\epsilon_2=0$. As revealed by the decodability in toy example, the idea that after the first phase, receiver ${\sf Rx}_2$ (the weaker receiver) decodes its message $W_2$. Since the first receiver has a stronger channel, it can recover interference $W_2$ in a shorter time horizon after the first segment of Phase~I. Then, during the second segment of Phase~I, the transmitter starts delivering cached $W_{1|2}$ to ${\sf Rx}_1$, while it continues delivering $W_2$ to ${\sf Rx}_2$. Note that since ${\sf Rx}_2$ knows $W_{1|2}$ and ${\sf Rx}_1$ has recovered $W_2$ in the first segment, the second segment benefits both receivers. Finally, during the newly-added second phase, $\bar{W}_{1|2}$ (the non-cached part of $W_1$ outside $W_{1|2}$) is delivered to ${\sf Rx}_1$. The whole process is summarized in Figure~\ref{Fig:Protocol-BICN}.\

\begin{figure}[!ht]
\centering
\includegraphics[width = 0.8\columnwidth]{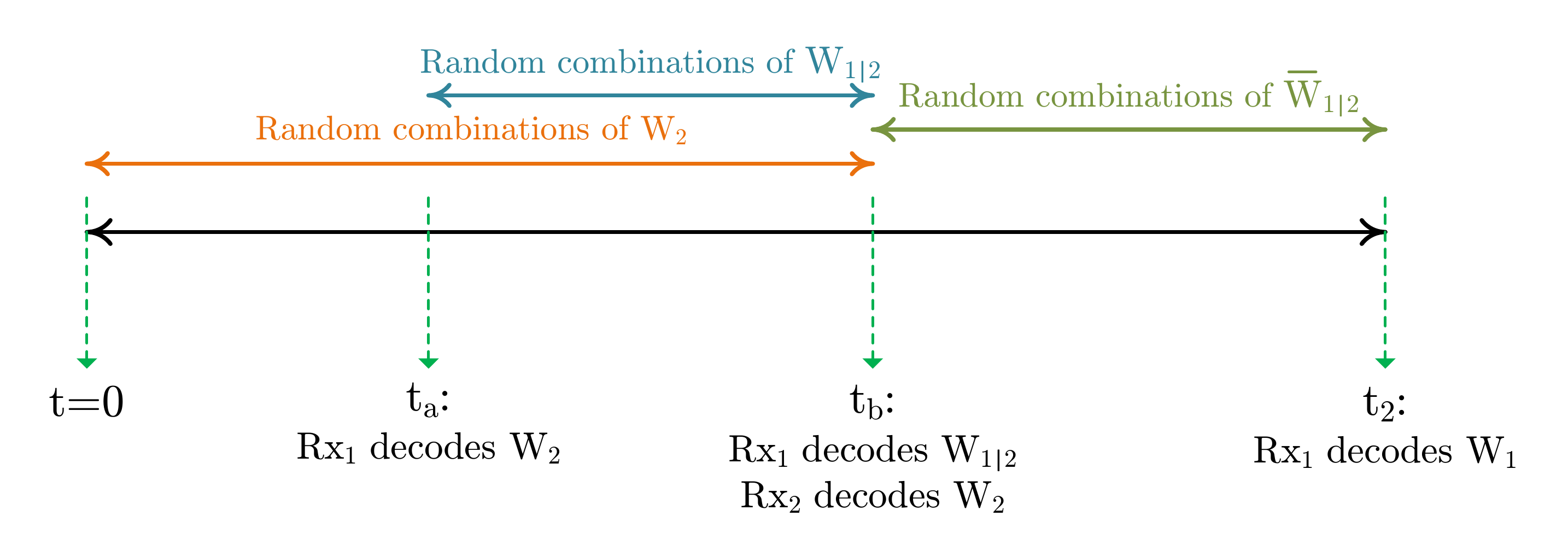}
\caption{Semi-blind two-phase protocol to achieve outer-bounds $\mathcal{C}^\mathrm{non-blind}_\mathrm{NN}$ in Theorem \ref{THM:Capacity_Out_BIC_No}.} \label{Fig:Protocol-BICN}
\end{figure}

\noindent {\bf Phase~I:} The transmitter creates
\begin{align}
\frac{m}{(1-\delta_2)} + O\left( m^{2/3} \right)
\end{align}
random linear combinations of the $m$ bits intended for ${\sf Rx}_2$ such that any randomly chosen $m$ combinations are randomly independent. This can be accomplished by a random linear codebook of which each element is generated from i.i.d. Bernoulli $\mathrm{Ber}(1/2)$. In practice, Fountain codes can be used.

\begin{remark}[Expected values of concentration results]
The $O\left( m^{2/3} \right)$ is to ensure sufficient number of equations will be received given the stochastic nature of the channel. At the end we let $m \rightarrow \infty$, such terms do not affect the achievability of the overall rates. Thus, for simplicity of expressions, we omit these terms and only work with the expected value of the number of equations in what follows, since the actual number will converge to this expection.
\end{remark}

\noindent {\bf Segment~a of Phase~I:} This segment has a total (on average) length of
$t_a$ as in \eqref{eq_NNta}. During this segment, the transmitter sends $t_a$ of the random combinations it created for ${\sf Rx}_2$ as described above and illustrated in Figure~\ref{Fig:Protocol-BICN}.

\noindent {\bf Segment~b of Phase~I:} This segment has a total (on average) length of
$t_b$ as \eqref{eq_NNtb}. During this segment, the transmitter creates $t_b$ random linear combinations of the $(1-\epsilon_2) \eta m$ bits in $W_{1|2}$, and XORs these combinations with the another $t_b$ random combinations it created for $W_2$. The transmitter sends the XORed sequence during Segment~b of Phase~I.

\noindent {\bf Phase~II:} This phase has a total (on average) length of
\begin{align}
t_2 = \frac{\epsilon_2\left( (1-\delta_1)-\epsilon_1(1-\delta_2) \right)}{(1-\epsilon_2)(1-\delta_1)(1-\delta_2)}m = \frac{\epsilon_2 \eta m}{(1-\delta_1)}. \label{eq_NNt2}
\end{align}
During this phase, the transmitter creates $t_2$ random linear combinations of the bits in $\bar{W}_{1|2}$ and sends these combinations.

The decodability for ${\sf Rx}_1$ comes as follows. At the end of the second phase, ${\sf Rx}_1$ gathers $\epsilon_2 \eta m$ linearly independent equations of non-cached $\bar{W}_{1|2}$, and as $\left| \bar{W}_{1|2} \right| = \epsilon_2 \eta m$, ${\sf Rx}_1$ can decode $\bar{W}_{1|2}$. For cached $W_{1|2}$, as in toy example, during segment b of Phase I, ${\sf Rx}_1$ can remove the interference ($W_2$ is known in segment a) and decode $W_{1|2}$. Thus, ${\sf Rx}_1$ can decode ${W}_{1|2}$ and $\bar{W}_{1|2}$, meaning that it can recover $W_1$. The decodability for ${\sf Rx}_2$ after Phase I  directly follows from that in the toy example.

\noindent {\bf Achievable rates:} Recall the total length $t_a+t_b$ of Phase~I is $\frac{m}{(1-\delta_2)}$ from \eqref{eq_NNtb}, then with \eqref{eq_NNt2}, the total communication time is given by:
\begin{align}
t_a+t_b + t_2 =  \frac{(1-\delta_1)-\epsilon_1 \epsilon_2 (1-\delta_2)}{(1-\epsilon_2)(1-\delta_1)(1-\delta_2)}m,
\end{align}
which immediately implies the achievability of rates in \eqref{Eq:GeneralRates}.

%%%%%%%%%%%%%%%%%%%%%%%%%%%%%%%%%%%%%%%%%%%%%%%%%%%
\section{Achievability Proof for $\mathcal{C}^\mathrm{blind}_\mathrm{DD}$ in Case A of Theorem~\ref{THM:Capacity_Out_BIC_Delayed}}
\label{Section:Achievability_BIC_DDSymmetric}

First, we remind the reader of the conditions in Case A of Theorem~\ref{THM:Capacity_Out_BIC_Delayed}: global delayed CSIT and symmetric channel where $\delta_1 = \delta_2 = \delta$ and $\epsilon_1 = \epsilon_2 = \epsilon$. For this setup, we present an opportunistic communication protocol that enables the transmitter to use the available delayed CSI and the statistical knowledge of  the available side-information at the receivers. This protocol starts by sending the summation (\emph{i.e.} XOR in the binary field) of individual bits intended for the two receivers, and then, based on the available channel feedback and the statistics of receiver side-information, the transmitter is able to efficiently create recycled bits for retransmission. In this regard, the protocol has some similarities to the reverse Maddah-Ali-Tse scheme~\cite{yang2012degrees}, which was originally designed for multiple-input fading broadcast channels~\cite{maddah2012completely}. As discussed in \cite{sclin2018IT}, the channel setting is fundamentally different: discrete erasure channel vs. continuous Rayleigh channel, single antenna vs. multiple-input transmitter. Together with the blind receiver side-information in our case; thus, our protocols end up sharing little ingredients with \cite{yang2012degrees}.

We skip the protocol for achieving $R_ i = \lp 1 - \delta \rp$, as it is a well-established result. Instead, we provide the achievability protocol for the maximum symmetric sum-rate point given by
\begin{align}
\label{Eq:MaxSumRate}
R_1 = R_2 = \frac{1-\delta^2}{1+\delta+\epsilon}.
\end{align}

\begin{figure*}[!t]
\centering
\includegraphics[width = \textwidth]{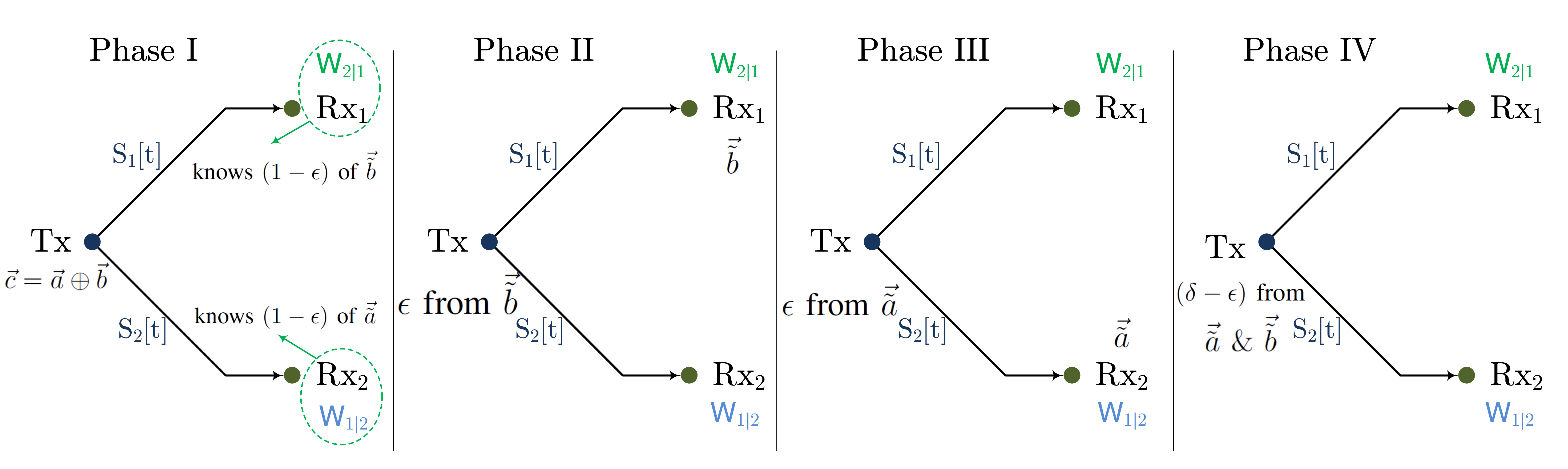}
\caption{The proposed four-phase protocol when $\epsilon \leq \delta$: Phase I delivers combinations of $\vec{a}$ and $\vec{b}$; Phases II and III deliver interfering bits to unintended receivers alongside useful information to the intended receivers; Phase IV multicasts XOR of available bits at each receiver needed at the other.\label{Fig:Protocol-BIC}}
\end{figure*}

We break the scheme based on the relationship between $\epsilon$ and $\delta$. We note that since we focus on the homogeneous setting in this section, the transmitter has $m$ bits for each receiver: $a_j$'s for receiver ${\sf Rx}_1$ and $b_j$'s for receiver ${\sf Rx}_2$, $j=1,2,\ldots,m$. % Also, as we focus on the binary field, we {\it refer to packets as bits}.

\noindent \underline{Scenario~1 ($\epsilon \leq \delta$)}: This case assumes the side channel that provides each receiver its side-information is stronger than the channel from the transmitter. The transmitter first creates $\vec{c} = \left(c_1,c_2,\ldots,c_{m} \right)$ according to
\begin{align}
c_j = a_j \oplus b_j, \qquad j=1,2,\ldots,m.
\end{align}
The protocol is divided into four phases as described below and depicted in Figure~\ref{Fig:Protocol-BIC}. \\
\noindent{\bf Phase~I}: During this phase, the transmitter sends out individual bits from $\vec{c}$ until at least one receiver obtains this bit, and then, moves on to the next bit. Due to the statistics of the channel, this phase takes on average
\begin{align}
t_1 = \frac{m}{1-\delta^2}
\end{align}
time instants.
\begin{remark}
To keep the description of the protocol simple, we use the expected value of different random variable (\emph{e.g.} length of phase, number of bits received by each user, etc). A more precise statement would use a concentration theorem result such as the Bernstein inequality to show the omitted terms do not affect the overall result and the achievable rates as done in~\cite{AlirezaBFICDelayed,vahid2018throughput}. Moreover, when talking about the number of bits or time instants, we are limited to integer numbers. If a ratio is not an integer number, we can use $\lceil \cdot \rceil$, the ceiling function, and since at the end we take the limit for $m \rightarrow \infty$, the results remain unchanged.
\end{remark}
After Phase~I is completed, receiver $\msf{Rx}_1$ obtains on average $m/(1+\delta)$ bits from $\vec{c}$. The transmitter, using channel feedback during Phase~I, knows which bits out of $\vec{b}$ where among those received by $\msf{Rx}_1$ as part of $\vec{c}$, denoted by $\vec{\tilde{b}}$ in Figure~\ref{Fig:Protocol-BIC}. Furthermore, $\msf{Rx}_1$ statistically knows a fraction $(1-\epsilon)$ of the interfering $\vec{\tilde{b}}$ from its side-information. Thus, if $\msf{Rx}_1$ obtains an additional fraction $\epsilon$ of $\vec{\tilde{b}}$, it can resolve interference in Phase~I to get $m/(1+\delta)$ {\it pure} bits from $\vec{a}$. A similar statement holds for $\msf{Rx}_2$. % the $\epsilon m/(1+\delta)$ bits of $\vec{b}$ it received as part of the XORed bits during Phase~I,

\noindent{\bf Phase~II}: The transmitter creates $\epsilon m/(1+\delta)$ linearly independent combinations of $\vec{\tilde{b}}$, and encodes them using an erasure code of rate $(1-\delta)$ and sends them out. This phase takes
\begin{align}
t_2 = \frac{\epsilon m}{1-\delta^2},
\end{align}
time instant, and upon its completion, $\msf{Rx}_1$ gets the additional equations to remove interference during Phase~I and recover $m/(1+\delta)$ bits from $\vec{a}$, while $\msf{Rx}_2$ obtains further $\epsilon m/(1+\delta)$ equations of its intended bits.

\noindent{\bf Phase~III}: This phase is similar to Phase~II, but the transmitter sends out $\vec{\tilde{a}},$ those bits out of $\vec{a}$ that were received by $\msf{Rx}_2$ as part of $\vec{c}$. This phase takes
\begin{align}
t_3 = \frac{\epsilon m}{1-\delta^2},
\end{align}
time instants, and upon its completion, $\msf{Rx}_2$ gets the additional equations to remove interference during Phase~I and recover $m/(1+\delta)$ bits from $\vec{b}$, while $\msf{Rx}_1$ obtains further $\epsilon m/(1+\delta)$ equations of its intended bits.

\noindent{\bf Number of equations at each receiver}: After the first three phases, each receiver has a total of
\begin{align}
\frac{m}{1+\delta} + \frac{\epsilon m}{1+\delta}
\end{align}
linearly independent equations of its bits, and thus, needs an additional
\begin{align}
\label{Eq:finalphasefraction}
\frac{\delta- \epsilon}{1+\delta}m
\end{align}
new equations to complete recovery of its bits. Note that if $\epsilon = \delta$, the protocol ends here.

 \noindent{\bf Phase~IV}: The transmitter creates
 \begin{align}
 \frac{\delta- \epsilon}{1+\delta}m
 \end{align}
 further linearly independent random combinations of the bits intended for $\msf{Rx}_i$ but received at $\msf{Rx}_{\bar{i}}$ as part of $\vec{c}$ during Phase~I, denoted by $\vec{\tilde{a}}$ and $\vec{\tilde{b}}$ in Figure~\ref{Fig:Protocol-BIC}. Note that at this point, each receiver has full knowledge of the interfering bits during Phase~I and retransmission of such bits will no longer create any interference. Thus, the transmitter encodes these two sets of bits (one for each receiver) using an erasure code of rate $\lp 1 - \delta \rp$ and send the XOR of these encoded bits. This phase takes
\begin{align}
t_4 = \frac{\delta- \epsilon}{1-\delta^2}m
\end{align}
time instants.

We note that in Phases~II,~III,~and~IV, the transmitter needs to create linearly independent combinations of the bits. Thus, we need to guarantee the feasibility of these operations. In Phase~I, as part of $\vec{c}$, a total of
\begin{align}
\label{Eq:knownatother}
\frac{\delta}{1+\delta}m
\end{align}
bits intended for one receiver arrive at the unintended receiver and effectively, during the next phases, we deliver these bits to the intended receiver. In fact, \eqref{Eq:knownatother} equals the summation of the number of linearly equations needed for $\msf{Rx}_1$ during Phases~III and~IV, and for $\msf{Rx}_2$ during Phases~II and~IV. Thus, the feasibility of creating sufficient number of linearly independent combinations is guaranteed.

Upon completion of Phase~IV, each receiver first removes the contribution of the bits intended for the other user, and then, recovers the additional equations needed as indicated in \eqref{Eq:finalphasefraction}, and thus, is able to complete recovery of its message.

\noindent{\bf Achievable rates}: The total communication time is
\begin{align}
\label{Eq:totaltime}
\sum_{j=1}^{4}{t_j} = \frac{1+\delta+\epsilon}{1-\delta^2}m,
\end{align}
which immediately results in target rates of \eqref{Eq:MaxSumRate}.

\noindent \underline{Scenario~2 ($\epsilon > \delta$)}: This scenario corresponds to the case in which the side channel that provides each receiver with its side-information is weaker than the channel from the transmitter. The protocol has four phases as before with some modifications. Phase~I remains identical to the previous scenario; during Phases~II and~III, instead of $\epsilon m/(1+\delta)$, the transmitter creates $\delta m/(1+\delta)$ linearly independent equations of $\vec{\tilde{b}}$ and $\vec{\tilde{a}}$, respectively, and sends them out as done in the previous scenario. With these modifications, after the first three phases, $\msf{Rx}_1$, still has
\begin{align}
\label{Eq:Sc2AfterPh12}
\frac{\epsilon - \delta}{1+\delta}m
\end{align}
equations interfered by $\vec{\tilde{b}}$. Thus, for $\msf{Rx}_1$ to successfully recover $\vec{a}$, the transmitter has two options: $(i)$ to deliver the same number as in $\eqref{Eq:Sc2AfterPh12}$, new combinations of $\vec{\tilde{b}}$ to $\msf{Rx}_1$, and $(ii)$ to provide $\msf{Rx}_1$ with additional combinations, same as in $\eqref{Eq:Sc2AfterPh12}$ , of its own bits $\vec{\tilde{a}}$. With the first choice, receiver $\msf{Rx}_1$ fully resolves the interference and recovers its bits; while with the second choice, it simply obtains $m$ linearly independent equations of $\vec{a}$. Interestingly, either choice is also good for $\msf{Rx}_2$: with the first choice, $\msf{Rx}_2$ obtains $m$ linearly independent equations of $\vec{b}$; while with the first choice, $\msf{Rx}_2$ fully resolves the interference and recovers its bits. In other words, in this scenario, during Phase~IV, no XOR operation is needed and only bits intended for one user would enable both receivers to decode their bits.
%to resolve the interfering $\vec{\tilde{b}}$ during Phase~I, and if that happens, $\msf{Rx}_1$ will recover all of its intended bits since it will have $m/(1+\delta)+\delta m/(1+\delta)=m$ linearly independent equations of $\vec{a}$. Alternatively, the transmitter can provide the same number as in \eqref{Eq:Sc2AfterPh12} further linearly independent equations of $\vec{a}$ in a future phase to help $\msf{Rx}_1$ decode its message. Similar statements hold for $\msf{Rx}_2$ after the first three phases.

Based on the discussion above, during Phase~IV, the transmitter creates
$
 \lp \epsilon - \delta \rp / \lp 1+\delta \rp m
$
linearly independent combinations, same as \eqref{Eq:Sc2AfterPh12}, of the bits in $\vec{\tilde{a}}$ intended for $\msf{Rx}_1$ but received as part of $\vec{c}$ at  $\msf{Rx}_2$ during Phase~I. Then, the transmitter encodes these equations using an erasure code of rate $\lp 1-\delta \rp$ and sends them to both receivers. Similar to the previous scenario, we can guarantee the feasibility of creating these linearly independent equations.
%linearly independent combinations, same as \eqref{Eq:Sc2AfterPh12}, of the bits in $\vec{\tilde{a}}$ intended for $\msf{Rx}_1$ but received as part of $\vec{c}$ at  $\msf{Rx}_2$ during Phase~I. {\red Similar combinations are created for $\red \vec{\tilde{a}}$. Then, the transmitter encodes XOR of these combinations} using an erasure code of rate $\lp 1-\delta \rp$ and sends them to both receivers. Similar to the previous scenario, we can guarantee the feasibility of creating these linearly independent equations.

After Phase~IV, as discussed above, $\msf{Rx}_1$ will have sufficient number of equations to recover $\vec{a}$; while $\msf{Rx}_2$ first needs to resolve the interference using the equations it obtains during this phase, and then, recover $\vec{b}$. We note that, the transmitter could send combinations of $\vec{b}$ instead of $\vec{a}$ during the last phase, and the decoding strategy of the receivers would get swapped.

\section{{Achievability Proofs} for $\mathcal{C}^{\mathrm{semi-blind}}_{\mathrm{DN}}$ in Case B and \\ $\mathcal{C}^{\mathrm{non-blind}}_{\mathrm{DN}}$ in Case C of Theorem~\ref{THM:Capacity_Out_BIC_Delayed}}
\label{Section:Section:Achievability_BIC_DN}

\newcommand{\E}{\mathbb{E}}

\noindent \textbf{Case B : Achievability for $\mathcal{C}^{\mathrm{semi-blind}}_{\mathrm{DN}}$ when $\epsilon_1 = 0$} : Now $\epsilon_1=0$ and thus $W_{2|1}=W_2$. Recall message $W_2$ is also represented by
a bit vector $\vec{b}$. The encoding process come as follows. At time index $t$, the $j$-th bit $a_{j}$ in message $\vec{a}$ for $\msf{Rx}_1$ is repeated according to the state feedback $S_1$, and after XORing a random linear combination $(\vec{g}_t)^\intercal \vec{b}$, the resulting superposition is sent. Each entry of $\vec{g}_t$ is generated from $i.i.d.$ $\mathrm{Ber}(1/2)$. Each $a_{j}$ is repeated until the corresponding state feedback is $S_1=1$. In other words, prior to the superposition via XORing, $\vec{a}$ is pre-encoded and repeated as in standard ARQ, while $\vec{b}$ is pre-encoded by a fountain-like random linear code. the termination of this fountain code is determined by the state feedback of $\msf{Rx}1$, that is whether or not all bits in $\vec{a}$ are successfully delivered to $\msf{Rx}1$.

The decoding at $\msf{Rx}_1$ follows from the standard ARQ, since full side-information
$\vec{b}$ is known. By
setting total transmission length as
\begin{equation} \label{eq_DNLength}
\frac{m_1}{1-\delta_1},
\end{equation}
the achievable rate is
\begin{equation} \label{eq_DNR1}
R_1=1-\delta_1
\end{equation}
For user $\msf{Rx}_2$, it has side-information $W_{1|2}$ for $(1-\epsilon_2)$ bits of the interference $\vec{a}$ and each reception of the corresponding XOR transmitted will result in a linear equation of $\vec{b}$. To see this, consider the $j$-th bit $a_j$ of interference $\vec{a}$. Suppose it is repeated $L_i$ times until its mixture with $\vec{b}$ is successfully arrived at $\mathsf{Rx}1$. Within this span, $\msf{Rx}2$ gets
\begin{equation} \label{eq_withSIclean}
K_i \triangleq \sum_{\ell=1}^{L_j}S_{2,j}[ \ell]
\end{equation}
linear equations mixing $a_j$ and $\vec{b}$, where $S_{2,j}[\ell]$ is the erasure state at $\msf{Rx}2$ for the $\ell$-th transmission of $a_j$. Then $\msf{Rx}2$ gets $K_i$ pure equations of $\vec{b}$. In total, the number of pure linear equations of message $\vec{b}$ is
\begin{equation} \label{eq_DNside}
m_1(1-\epsilon_2) \mathbb{E}[K_i]= m_1(1-\epsilon_2) \frac{1-\delta_2}{1-\delta_1},
\end{equation}
For interference bits without side-information, by using interference alignment in \cite{lin2019no}, the number of pure linear equations is
\begin{equation}
\label{eq_DNnoside}
m_1\epsilon_2 \E[(K_i-1)^+]=m_1 \epsilon_2 \left( \frac{\delta_1-\delta_2}{1-\delta_1}+\frac{\delta_2-\delta_1\delta_2}{1-\delta_1\delta_2} \right)
\end{equation}
Thus the total number of pure equations from \eqref{eq_DNside} and \eqref{eq_DNnoside} is
\[
m_1(1-\delta_2)\lp \frac{1}{1-\delta_1}-\frac{\epsilon_2}{1-\delta_1\delta_2} \rp,
\]
which results in achievable rate
\begin{equation} \label{eq_DNR2}
R_2=(1-\delta_2)\lp 1-\epsilon_2 \frac{1-\delta_1}{1-\delta_1\delta_2} \rp
\end{equation}
where \eqref{eq_DNLength} is applied. It can be checked that $(R_1,R_2)$ in \eqref{eq_DNR1}\eqref{eq_DNR2} is the corner point of outer-bound region \eqref{Eq:Capacity_Out_BIC_Delayed} when $\epsilon_1=0$
\begin{align}
 \epsilon_2 \frac{1-\delta_2}{1-\delta_1\delta_2} & R_1 + R_2 \leq \lp 1 - \delta_2 \rp \\
 0 \leq & R_i \leq (1-\delta_i),   i = 1,2,
\end{align}
Other corner point can be trivially achieved by time-sharing.

\begin{figure*}
\includegraphics[width = \textwidth]{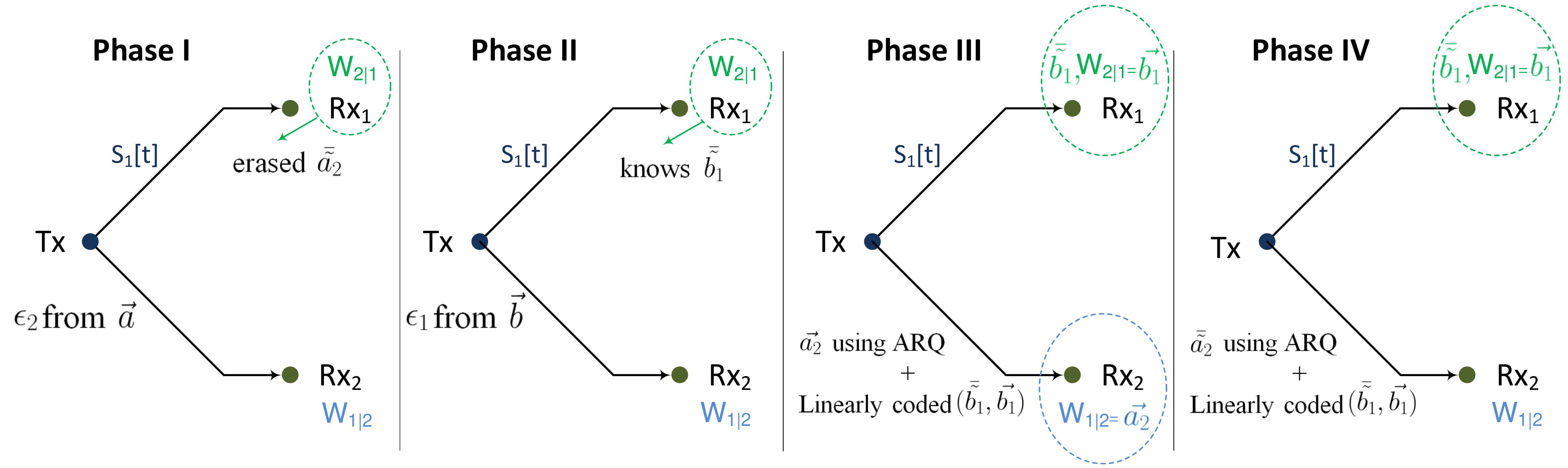}
\caption{Proposed four-phase protocol for achieving $\mathcal{C}^{\mathrm{non-blind}}_{\mathrm{DN}}$.}\label{Fig:Protocol-DNIC}
\end{figure*}

\noindent \textbf{Case C : Achievability for $\mathcal{C}^{\mathrm{non-blind}}_{\mathrm{DN}}$} : As in Fig.~\ref{Fig:Protocol-DNIC}, we now introduce the four-phase scheme for this achievability, of which the third and fourth phases are similar to that in Case B. We first represent $W_{2|1}$ and $W_{1|2}$ using bit vectors $\vec{b}_1$ and $\vec{a}_2$ respectively, then the encoding process is \\
\noindent \textbf{Phase I} : The transmitter sends bits from $\vec{a}$ which are not cached at $\msf{Rx}_2$ and not in $\vec{a}_2$. The total length $t_1$ of Phase I is $\epsilon_2 m_1$. After Phase I, the transmitter knows length $t_1\delta_1$ sequence $\bar{\tilde{a}}_2$, which is formed by bits erased at $\msf{Rx}1$ in Phase I where $S_1[t]=0$.

\noindent \textbf{Phase II} : The transmitter selects $\epsilon_1 m_2$ bits from $\vec{b}$ which are not cached at $\msf{Rx}_1$ and not in $\vec{b}_1$, and send then random linear combinations of them.  The total length $t_2$ of Phase II is $\epsilon_1 m_2/(1-\delta_1\delta_2)$. After Phase II, the transmitter knows length $t_2(1-\delta_1)$ sequence $\bar{\tilde{b}}_1$, which is formed by bits received at $\msf{Rx}1$ in Phase I where $S_1[t]=1$.

\noindent \textbf{Phase III} : The transmission is similar to that in semi-blind Case B, the differences are as follows. Now the transmitter is non-blind to $\vec{a}$ so it pre-encodes cached $\vec{a}_2$ instead of the whole $\vec{a}$ using ARQ. Also the transmitter is only sure that  $(\bar{\tilde{b}}_1,\vec{b}_1)$ is known at $\msf{Rx}_1$, it pre-encodes $(\bar{\tilde{b}}_1,\vec{b}_1)$ instead of full $\vec{b}$ using the random linear code. More specifically, the output of the second pre-encoder at time $t$ becomes the XOR of random linear combinations $(\vec{g}_t)^\intercal \bar{\tilde{b}}_1 \oplus (\vec{g'}_t)^\intercal \vec{b}_1$, where each entry of $\vec{g}_t$ or $\vec{g'}_t$ is generated from i.i.d. $\mathrm{Ber}(1/2).$ For the first pre-encoder,  each bit in $\vec{a}_2$ cached at user 2 is
repeated according to the delayed $S_1$ as described in Case B. Finally the XOR of outputs of these two pre-encoders is sent.

\noindent \textbf{Phase IV}: The transmission in phase is same as that in Phase III, by replacing the input of the first ARQ pre-encoder by the recycled $\bar{\tilde{a}}_2$. Though $(1-\delta_2)$ of sequence $\bar{\tilde{a}}_2$ will be known at $\msf{Rx}_2$ in Phase IV, the transmitter is blind to these bits. On the contrary, in Phase III the transmitter knows that the input $\vec{a}_2$ of the ARQ pre-encoder is totally cached at $\msf{Rx}_2$.

\noindent Note that without receiver side-information $\epsilon_1 = \epsilon_2 =1$, there will be no Phase III and our scheme reduces to the three-phase scheme in \cite{lin2019no}.

We focus on the decodability for receiver $\msf{Rx}1$ first. In Phase III and IV, since $(\bar{\tilde{b}}_1,\vec{b}_1)$ is already known at $\msf{Rx}_1$, $\vec{a}_2$ and $\bar{\tilde{a}}_2$ can be recovered, if the lengthes of Phases are respectively chosen as
\begin{align}
t_3 (1-\delta_1) = (1-\epsilon_2) m_1 \notag \\
t_4 (1-\delta_1) = \delta_1 t_1 = \delta_1 \epsilon_2 m_1 \label{eq_DNcachet34}
\end{align}
Together with bits received in Phase I, receiver $\msf{Rx}1$ gets
all $m_1$ bits. Now, we turn to the decodability at $\msf{Rx}2$.
Receiver $\msf{Rx}2$ will first decode
super-$(\bar{\tilde{b}}_1,\vec{b}_1)$ from its received bits
during the entire three phases. With the $\bar{\tilde{b}}_1$ and
the received equations during Phase II, it will have
$t_2(1-\delta_1\delta_2)=\epsilon_1 m_2$ equations to decode
uncached bits in $\vec{b}$. Together with $\vec{b}_1$, the whole
message for user 2 is decoded. To ensure successful decoding of
the super-$(\bar{\tilde{b}}_1,\vec{b}_1)$, we calculate the
corresponding expected number of linearly independent equations as
follows. In Phase III, every reception at $\msf{Rx}_2$ will result
in a new equation since $\vec{a}_2$ is cached, and we have
\[
(1-\epsilon_2)m_1 \E[K_i]
\]
equations after Phase III. In Phase IV, as \eqref{eq_DNside} and
\eqref{eq_DNnoside}, we will have additional
\[
t_1 \delta_1(1-\delta_2) \E[K_i]+ t_1 \delta_1\delta_2
\E[(K_i-1)^+]
\]
equations since $(1-\delta_2)$ of $\bar{\tilde{a}}_2$ will be
received during Phase I. By using equations of $\bar{\tilde{b}}_1$
received at $\msf{Rx}_2$ in Phase II as additional cache, we need
\begin{align}
&t_2(1-\delta_1)+(1-\epsilon_1)m_2  \leq \notag \\
&t_2(1-\delta_1-\delta_2+\delta_{1}\delta_2)+ (1-\epsilon_2)m_1
\frac{1-\delta_2}{1-\delta_1}+ t_1 \delta_1(1-\delta_2) \left(
\frac{1}{1-\delta_1}-\frac{\delta_2}{1-\delta_1\delta_2} \right)
\label{eq_DNcacheRx2Dec}
\end{align}
for successful decoding the length
$t_2(1-\delta_1)+(1-\epsilon_1)m_2$
super-$(\bar{\tilde{b}}_1,\vec{b}_1)$. Note that by collecting
$\bar{\tilde{b}}_1$ received in Phase II (with standard basis for
$\bar{\tilde{b}}_1$) and the linear equations produced in Phase III
and IV, one can form a set of linear equations of
$(\bar{\tilde{b}}_1,\vec{b}_1)$ described by a full (column) rank
matrix, when codelengths are long enough.

For $(R_1,R_2)$ satisfying outer-bound $R_1/(1-\delta_1)+\epsilon_1 R_2/(1-\delta_1\delta_2)=1$ in
\eqref{Eq:Capacity_Out_BIC_Delayed}, the total communication time must meet
\[
\sum^4_{j=1} t_j=\frac{m_1}{1-\delta_1}+\frac{\epsilon_1 m_2}{1-\delta_1\delta_2}.
\]
From selected lengths of Phase I and II, $m_1=t_1/\epsilon_2$ and $m_2=t_2(1-\delta_1\delta_2)/\epsilon_1$, together with \eqref{eq_DNcachet34}, this constraint is meet since
\[
\sum^4_{j=1} t_j =t_1 \left(1+\frac{(1-\epsilon_2)/\epsilon_2+\delta_1}{1-\delta_1}\right)+t_2=\frac{t_1}{\epsilon_2(1-\delta_1)}+t_2
\]
For the corner point $(R_1,R_2)$ which also satisfies outer-bound $R_2/(1-\delta_2)+\epsilon_2 R_1/(1-\delta_1\delta_2)=1$ in
\eqref{Eq:Capacity_Out_BIC_Delayed}, we further show that the decodability \eqref{eq_DNcacheRx2Dec} will also be met. From \eqref{Eq:Capacity_Out_BIC_Delayed},
\[
\frac{t_1}{1-\delta_1\delta_2}+\frac{1-\delta_1\delta_2}{\epsilon_1(1-\delta_2)}t_2= \sum^4_{j=1} t_j =\frac{t_1}{\epsilon_2(1-\delta_1)}+t_2,
\]
which implies
\[
t_2 \left( \frac{1-\delta_1\delta_2}{\epsilon_1}-(1-\delta_2)\right)=t_1 (1-\delta_2)\left( \frac{1}{\epsilon_2(1-\delta_1)}- \frac{1}{1-\delta_1\delta_2}\right),
\]
or equivalently
\[
t_2 \left(\delta_2-\delta_1\delta_2+\frac{1-\epsilon_1}{\epsilon_1}(1-\delta_1\delta_2)\right)=t_1 (1-\delta_2)\left( \frac{(1-\epsilon_2)/\epsilon_2}{1-\delta_1}+\left(\frac{1}{1-\delta_1}-1\right)- \left(\frac{1}{1-\delta_1\delta_2}-1\right)\right).
\]
 Then \eqref{eq_DNcacheRx2Dec} is met since $m_1=t_1/\epsilon_2$ and $m_2=t_2(1-\delta_1\delta_2)/\epsilon_1$.

%%%%%%%%%%%%%%%%%%%%%%%%%%%%%%%%%%%
%%%%%%%%%%%%%%%%%%%%%%%%%%%%%%%%%%%

\section{Conclusion}
\label{Section:Conclusion_BIC}

We studied the problem of communications over two-user broadcast erasure channels with random receiver side-information. We assumed the transmitter may not have access to global channel state information and global cache index information for both receivers. For the non-blind-transmitter case, we characterized the capacity region, while with a blind transmitter we showed the outer-bounds can be achieved under certain conditions. Thus, in general with a blind transmitter, the capacity region of the problem, also known as blind index coding over the broadcast erasure channel, remains open.

\appendices

\section{Proof of \eqref{Eq:withSideInfoN} in Lemma \ref{Lemma:Leakage_BIC_No}} \label{AppEqwithSideInfoN}
For time instant $t$ where $2 \leq t \leq n$, we have
\begin{align}
& H\left( Y_2[t] | Y_2^{t-1}, W_{1|2}, W_2, S_2^t, E^n \right) \\
& \quad = \lp 1 - \delta_2 \rp H\left( X[t] | Y_2^{t-1}, W_{1|2}, W_2, S_{2}[t] = 1,  S_2^{t-1}, E^n \right) \nonumber \\
& \quad \overset{(a)}= \lp 1 - \delta_2 \rp H\left( X[t] | Y_2^{t-1}, W_{1|2}, W_2, S_2^t, E^n \right) \nonumber \\
& \quad \overset{(b)}\geq \lp 1 - \delta_2 \rp H\left( X[t] | Y_1^{t-1}, W_{1|2}, W_2, S_1^t, E^n \right) \nonumber \\
& \quad \overset{(c)}= \frac{1 - \delta_2}{1 - \delta_1} H\left( Y_1[t] | Y_1^{t-1}, W_{1|2}, W_2, S_1^t, E^n \right),
\end{align}
where $(a)$ holds since $X[t]$ is independent of the channel realization at time instant $t$; $(b)$ follows from the following arguments. Consider a virtual channel state $\tilde{S}[t]$ is an i.i.d. Bernoulli $(\delta_2-\delta_1)/\delta_2$ process independent of $S_2[t]$ and transmitted signal, then, we have
\begin{align}
& H\left( X[t] | Y_2^{t-1}, W_{1|2}, W_2, S_2^t , E^n \right) \nonumber \\
&= H\left( X[t] | Y_2^{t-1}, W_{1|2}, W_2, S_2^t , \tilde{S}^{t-1} , E^n \right) \nonumber \\
&\geq  H\left( X[t] | \{ (1-S_2[\ell])\tilde{S}[\ell]X[\ell] \}_{\ell= 1}^{\ell = t-1}, Y_2^{t-1}, W_{1|2}, W_2, S_2^t , \tilde{S}^{t-1} , E^n \right) \nonumber \\
&= H\left( X[t] | Y_1^{t-1}, W_{1|2}, W_2, S_1^t , E^n \right),
\end{align}
where the third equality holds since receiving both virtual $(1-S_2[\ell])\tilde{S}[\ell]X[\ell]$ and $Y_2[\ell]=S_2[\ell]X[\ell]$ is statistically the same as receiving $S_1[\ell]X[\ell]=Y_1[\ell]$ (note if $S_2[\ell]=0$ there is $(\delta_2-\delta_1)/\delta_2$ probability $(1-S_2[\ell])\tilde{S}[\ell]=1$), and $X[t]$ is independent of the channel states and virtual $\tilde{S}^{t-1}$; $(c)$ comes from the fact that $\Pr\lp S_1[t]  = 0 \rp = \delta_1$. Next, taking the summation over $t$ from $1$ to $n$, and using the fact that the transmit signal at time instant $t$ is independent of future channel realizations, we get
\[
H\left( Y_2^n | W_{1|2}, W_2, S^n_2, E^n \right)  \geq \frac{1-\delta_2}{1-\delta_1} H\left( Y_1^n | W_{1|2}, W_2, S^n_1, E^n \right)
\]
and also \eqref{Eq:withSideInfoN}.

\section{Proof of $\mathcal{C}^\mathrm{blind}_\mathrm{NN}$ in Theorem~\ref{THM:Blind}: Opportunistic Transmission}
\label{Section:Achievability_Blind}

\noindent {\bf Case~1 : $\delta_2 \geq \delta_1$ and $\epsilon_2 \in \{ 0, 1 \}$}: First, we note that as stated in Remark~\ref{remark:weakvsstrong},  when the weaker receiver has no side-information, \emph{i.e.} $\epsilon_2 = 1$, then, the capacity region is the same as having no side-information at either receivers. The case with $\epsilon_2 = 0$ is already given in Section \ref{Section:Achievability_NonBlind}

\noindent {\bf Case~2 : $\delta_1 = \delta_2 = \delta$, and $\epsilon_1 = \epsilon_2 = \epsilon$ (symmetric setting)}:

Now we focus on the blind-transmitter assumption where the transmitter no longer has the luxury of knowing $W_{1|2}$ to send bits in such a way to benefit both receivers (as was done in Segment~b of Phase~I in the previous section \ref{Section:Achievability_NonBlind}). We note that if both $\epsilon_1$ and $\epsilon_2$ are equal to zero, the problem becomes trivial as each receiver has full side-information of the other user's message and $R_i = (1 - \delta_i)$ is achievable for $i = 1,2$.  For case $\epsilon_1 = \epsilon_2 = \epsilon >0$, if each receiver obtains a total of $(1+\epsilon)m$ linearly independent observation of both $W_1$ and $W_2$, then, it can recover both messages. It turns out that this protocol is capacity-achieving. The non-trivial corner point in this case is given by:
\begin{align}
\label{Eq:Symmetric_Rates}
R_1 = R_2 = \frac{(1-\delta)}{(1+\epsilon)}.
\end{align}
The protocol is straightforward. The transmitter starts with $m$ bits for each receiver and creates $(1+\epsilon)m$ linearly independent combinations of the total $2m$ bits for the two receivers. Then, the transmitter encodes these combinations using an erasure code of rate $(1-\delta)$ and communicates the encoded message. Each receiver by the end of the communication block will have $2m$ linearly independent observations of the $2m$ unknown variables and can decode both messages. This immediately implies the achievability of the corner point described in \eqref{Eq:Symmetric_Rates}.

\section{Proof of Theorem~\ref{THM:Blind-Ach}} \label{Section:Achievability_BlindInner}

In this case, receiver ${\sf Rx}_1$ (the stronger receiver) has full side-information of the message for ${\sf Rx}_2$, \emph{i.e.} $W_{2|1} = W_2$, and receiver ${\sf Rx}_2$ has access to $(1-\epsilon_2)$ of the bits intended for ${\sf Rx}_1$. The outer-bounds of Theorem~\ref{THM:Capacity_Out_BIC_No} in this case become:
\begin{equation}
\label{Eq:FullSideInfo_Stronger}
\left\{ \begin{array}{ll}
0 \leq \epsilon_2 \frac{1-\delta_2}{1-\delta_1} R_1 + R_2 \leq \left( 1 - \delta_2 \right), & \\
0 \leq R_1 \leq \left( 1 - \delta_1 \right).
\end{array} \right.
\end{equation}
Thus, the non-trivial corner point is given by:
\begin{align}
\label{Eq:FullSideInfo_Stronger_Rates}
&R_1 = (1-\delta_1), \nonumber \\
&R_2 = (1-\epsilon_2)(1-\delta_2).
\end{align}

In this case. we cannot achieve the corner point given in \eqref{Eq:FullSideInfo_Stronger_Rates}. To achieve the region described in Theorem~\ref{THM:Blind-Ach}, we need to prove the achievability of the following corner point:
\begin{align}
\label{Eq:FullSideInfo_Stronger_Rates_Inner}
&R_1 = (1-\delta_1), \nonumber \\
&R_2 = (1-\epsilon_2)(1-\delta_1)(1-\delta_2).
\end{align}

%In this case, we set
%\begin{align}
%\eta = \frac{R_1}{R_2} = \frac{(1-\delta_1)}{(1-\epsilon_2)(1-\delta_2)} > 1.
%\end{align}

\noindent {\bf Achievability protocol:} We start with $m$ bits for ${\sf Rx}_2$ and $\eta m$ bits for ${\sf Rx}_1$, where
\begin{align}
\eta = \frac{R_1}{R_2} = \frac{1}{(1-\epsilon_2)(1-\delta_2)} > 1.
\end{align}

The achievability protocol is carried over two phases. During the first phase, the transmitter creates $\eta m$ random linear combinations of the $m$ bits for ${\sf Rx}_2$ such that each subset of $m$ combinations are linearly independent. The transmitter then sends out the XOR of these combinations with the uncoded $\eta m$ bits of ${\sf Rx}_1$. Thus, this phase has a length of
\begin{align}
t_1 = \eta m.
\end{align}
During this phase, ${\sf Rx}_1$ obtains $(1-\delta_1)\eta m$ of its bits as it has access to $W_2$ as side-information and can cancel out the interference. Moreover,  ${\sf Rx}_2$ obtains $(1-\delta_2) \eta m$ XORed combinations, and since ${\sf Rx}_2$ statistically knows $(1-\epsilon_2)$ of the bits for ${\sf Rx}_1$, we conclude that ${\sf Rx}_2$ gathers
\begin{align}
(1-\epsilon_2)(1-\delta_2) \eta m = m
\end{align}
linearly independent combinations of its $m$ bits and is able to decode its message $W_2$.

The second phase has a total length of
\begin{align}
t_2 = \frac{\delta_1}{(1-\delta_1)}\eta m.
\end{align}
During this second phase, the transmitter creates $t_2$ random linear combinations of the $\eta m$ bits intended for ${\sf Rx}_1$ and sends them out. At the end of this phase, the first receiver gathers additional $(1-\delta_1)\eta m$ random equations of its intended bits and combined with the $(1-\delta_1)\eta m$ bits it already knows from the first phase, ${\sf Rx}_1$ is able to decode $W_1$.

\noindent {\bf Achievable rates:} The total communication time is
\begin{align}
t_1 + t_2 = \frac{\eta m}{(1-\delta_1)}.
\end{align}
This immediately implies the achievability of the rates given in \eqref{Eq:FullSideInfo_Stronger_Rates_Inner}.

\bibliographystyle{ieeetr}
\bibliography{bib_FBBudget,bibTWR}

\end{document}